\documentclass{sig-alternate-10pt}
\usepackage{lscape}
\usepackage{mathrsfs}
\usepackage{amsfonts}
\usepackage{bbm}
\usepackage{mathrsfs}
\usepackage{textcomp}
\usepackage{cite}
\usepackage{times}
\usepackage{graphicx}
\usepackage{amsfonts}
\usepackage{flushend}
\usepackage{algorithmic}
\usepackage[ruled]{algorithm}
\usepackage{cite}
\usepackage{times}
\usepackage{amsmath}
\usepackage{graphicx}
\usepackage{epsfig}
\usepackage{amsmath}
\usepackage{latexsym}
\usepackage{amsfonts}
\usepackage{amssymb}
\usepackage{paralist}
\usepackage{xspace}
\usepackage{mathrsfs}
\usepackage{amssymb}
\usepackage{color}
\usepackage{algorithmic}
\usepackage[ruled]{algorithm}
\usepackage{flushend}
\usepackage{subfigure}
\usepackage{url}
\newtheorem{thm}{Theorem~}
\newtheorem{lem}{Lemma~}
\newtheorem{defn}{Definition~}

\allowdisplaybreaks 

\renewcommand{\paragraph}[1]{\smallskip \noindent {\textsc{#1}}}

\def\ie{\textit{i.e.}\xspace}

\DeclareMathOperator{\EMST}{EMST}

\begin{document}
\title{Measuring Transport Difficulty of Data Dissemination in \\ Large-Scale Online Social Networks:\\ An Interest-Driven Case}
\numberofauthors{1}
\author{
\alignauthor
Cheng Wang, ~Zhenzhen Zhang, ~Jieren Zhou,  ~Yuan He, ~Jipeng Cui, ~Changjun Jiang\\
       \affaddr{Department of Computer Science and Technology}\\
       \affaddr{Tongji University, Shanghai, China}\\
       \email{chengwang@tongji.edu.cn}
}

\maketitle
\begin{abstract}
In this paper, we aim to model the formation of data dissemination in online social networks (OSNs), and measure the transport difficulty of generated data traffic.
We focus on a usual type of interest-driven social sessions in OSNs, called \emph{Social-InterestCast}, under which a user will autonomously determine whether to view the content from his followees depending on his interest.
It is challenging to figure out the formation mechanism
of such a  Social-InterestCast, since it involves multiple interrelated factors such as users' social relationships, users' interests, and content semantics.
We propose a four-layered system model, consisting of physical layer, social layer, content layer, and session layer.
By this model we successfully obtain the geographical distribution of Social-InterestCast sessions, serving as the precondition for quantifying data transport difficulty.
We define the fundamental limit of  \emph{transport load} as  a new metric, called \emph{transport complexity}, i.e., the \emph{minimum required} transport load for an OSN over a given carrier network.
Specifically, we derive the transport complexity
for Social-InterestCast sessions in a large-scale OSN over the carrier network with optimal communication architecture.
The results can act as the common lower bounds on transport load for Social-InterestCast over any carrier networks.
To the best of our knowledge, this is the first work to measure the transport difficulty for data dissemination in OSNs by modeling session patterns with the interest-driven characteristics.
\end{abstract}



\maketitle

\section{Introduction}\label{section-Introduction}

As social networking services are becoming increasingly popular, online social networks (OSNs) play a growing role in individuals' daily lives.
The user population of OSNs has grown drastically in recent years.
According to the research proceeded by Statista\footnote{One of the largest statistics portals, http://www.statista.com/}
demonstrated that the number of social network users worldwide in 2014 had reached 1.87 billion and estimated that there will be around 2.72 billion social network users around the globe in 2019 \cite{Statista}.
In addition, people are spending more and more time on OSNs. For example, in January 2015, GlobalWebIndex\footnote{A market research firm running world's largest market research study on the digital consumer, https://www.globalwebindex.net/}   showed that the average user spends 1.72 hours per day on social platforms, which represents about 28 percent of all online activity \cite{Globalwebindex}. What's more, various kinds of social applications are constantly emerging, which are rendering an increasing user population of OSNs and providing users with more types of content to choose from, e.g., audios, videos, and pictures.
In addition, OSNs are covering a wider range around the world. All of the factors involved above are resulting in the heavy load imposed on the carrier communication network of OSNs. Furthermore, the load in the OSNs will increase continually with the expansion of OSNs.

Over a long period of time, this growth will give rise to the limitation of Internet's bandwidth, so practically measuring the load imposed by OSNs enjoys a crucial meaning.
To measure such a load, i.e., the load imposed on the carrier communication network by the OSN, we first introduce a metric called \emph{transport load}. It is defined as the product of two key factors: data generating rate at users and transport distance of messages.
The former depends on the format of message (e.g., text, picture, and video, etc.) and the performance request in terms of quality of service (e.g., throughput, latency, and deliver ratio, etc.).
The latter is directly determined by the geographical distribution of dissemination sessions and the communication architecture of network.
Both parameters are critical for measuring the transport capacity of networks. To be specific, it is defined as the product of bits and the transport distance over which the data is successfully transported from the source to the intended destinations.
In this work, we further define the fundamental limit of such transport load as
\emph{transport complexity}, i.e., the \emph{minimum required} transport load for an OSN over a given carrier communication network.
We note that, unlike some classical performance metrics, e.g., network capacity,
the transport complexity is a metric to define the fundamental transport difficulty of a specific data communication applications instead of the transport capability of a given network for specific applications.
Regarding the essential difference between transport capacity of a network and  transport complexity of a data transport application, we give an explicit explanation with the help of the illustration in Figure \ref{fig-capacity-vs-complexity}.

\begin{figure}[t]
\begin{center}
\scalebox{0.6}{\includegraphics{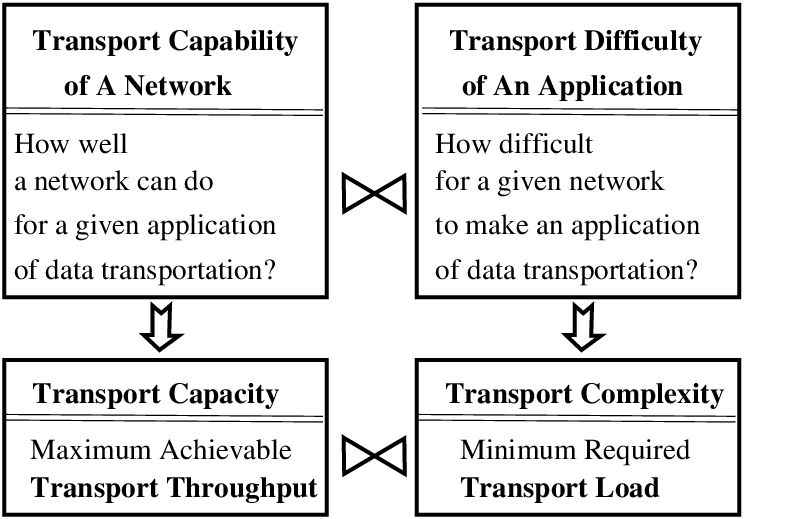}}
\caption{Transport Complexity vs. Transport Capacity. They are the typical and intuitive  metrics of  the transport difficulty of an application and
the transport capability of a network, respectively.
The transport capability is the capability of a network in terms of data transportation, while the transport difficulty is the inherent difficulty of an application in terms of data transportation.
} \label{fig-capacity-vs-complexity}
\end{center}
\vspace{-0.3in}
\end{figure}

For computing the transport load, the first thing is to comb rightly the procedure of a data dissemination in OSNs.
From our perspective, the session can be divided into two successive phases: \emph{Passive Phase} and \emph{Initiative Phase}.
In the passive phase, a source user holding a message sends  the digest of this message to all his followers.  The followers passively receive this digest.
This process of such a dissemination just acts like a \emph{broadcast} in all the followers of this source user.
We call this process \emph{Social-BroadCast}, and sessions generated in this process are straightforward called \emph{Social-BroadCast sessions}.
In the initiative phase, according to the interest, a user will autonomously determine whether to download the complete message  based on the digest of a message from who he follows, \cite{yan2014mining}.
We define such a dissemination process \emph{Social-InterestCast} and sessions generated in this process are \emph{Social-InterestCast sessions}.
In our work, the Social-InterestCast is defined by taking into account specific user interests and their effects on session generation. We give an intuitive explanation as follows:
There exists usually a \emph{user-message mapping} between the user set and message set, \cite{zhou2010impact}. That is,
when a user broadcasts a message to others, only the users whose interests are consistent with the topic of the message can be the potential destinations. For example, if a user broadcasts a video message characterized by several words (a digest or a title) in his Facebook, only the followers who are interested in this message will open the video file. It is convincing that this behavior occurs due to an underlying user-message mapping. In fact, the transport load in the former process, i.e., the transport load for the Social-BroadCast has been investigated in \cite{mobihoc2014}, then our focus of this work is to model the formation of a Social-InterestCast.
Accordingly, we propose
a four-layered model consisting of
the physical layer, social layer, content layer, and session layer, as illustrated in Figure \ref{fig-4-layer-archi}.
Compared with the three-layered model in \cite{mobihoc2014}, this model introduces a \emph{content layer} where the relationship links are defined as the semantic similarity among the  messages.
We take an example shown in Figure \ref{fig-4-layer-archi} to explain this procedure:
A user delivers a ``Message 1". Then, all his four followers can receive the glance (or say abstract) of ``Message 1".
Finally, only two followers are filtered through the content layer to act as the valid destinations due to their interests to ``Message 1".

After the preparations made above, we compute the transport load.
Recall its definition,
we first need to
 investigate the complex geographic characteristics of data dissemination sessions in OSNs, i.e., the spacial distribution of traffic sessions (the location distribution of sources and destinations).

We adopt the following three steps (as shown in  Figure \ref{model-steps}) to get the spacial distribution of Social-InterestCast sessions depending on users' geographical distribution:

$\rhd$ Firstly, we
model the spatial distribution of social relationships, i.e., the geographical distribution of followers,  by investigating the correlations between user's social relationship formation and geographical distribution.
We adopt the \emph{population-based model} in \cite{mobihoc2014} for the advantages in realistic and analytical aspects.

$\rhd$ Secondly, we build the user-content interest mapping by matching the topics of content and interests of users.
Under the Social-InterestCast, for measuring the transport complexity, it is important to estimate how many followers of a source user will be interested in a given message and make a decision to view it.
We divide this problem into two cases in terms of the dependency among followers' decisions and the attractivity difference of message content.
To be specific,
when the decisions of a source user's followers are non-independent, Matthew Effect \cite{merton1968matthew} should be considered due to the preferential attachment of some messages that have attracted more followers;
when the attractivity of messages is of significant difference, the preferential attachment of some popular messages comes on the stage. Furthermore, combining with the experimental results based on Foursquare users' dataset \cite{Bao2012Location}, we get that: Under both cases above, for a source user, say $v_{k}$, we reasonably assume that the number of destinations of a session from $v_{k}$ follows a Zipf's distribution whose parameters depend on the degree of node $v_{k}$.

$\rhd$ Thirdly, we model the spatial distribution of traffic sessions, i.e., the geographical distribution of session sources and destinations,
by combining the effects of both users' social relationships on the social layer and the user-content interest mapping across the social and  content layers.

\begin{figure}[t]
\begin{center}
\scalebox{0.45}{\includegraphics{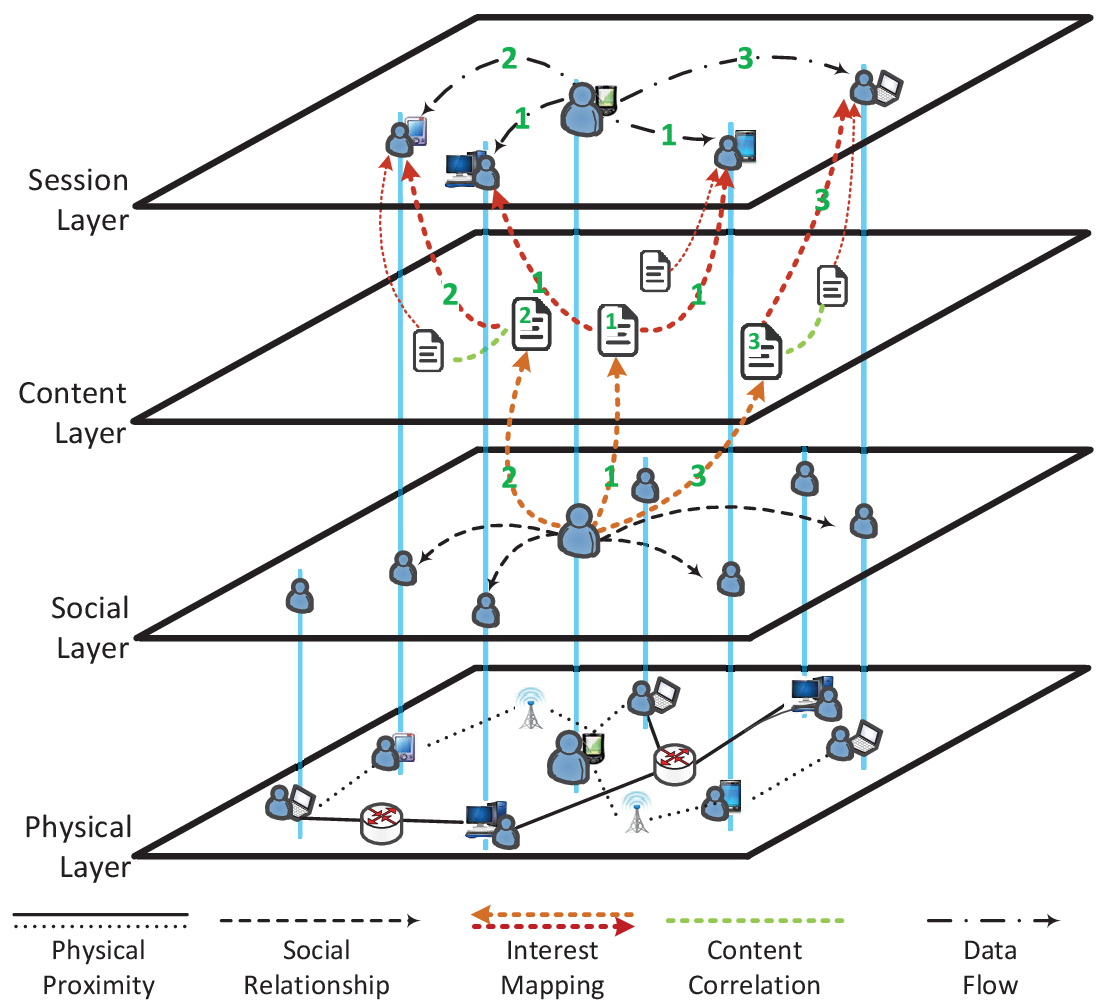}}
\caption{Four-layered Architecture and Social-InterestCast.}\label{fig-4-layer-archi}
\end{center}
\vspace{-0.2in}
\end{figure}

\begin{figure}[t]
\begin{center}
\scalebox{0.4}{\includegraphics{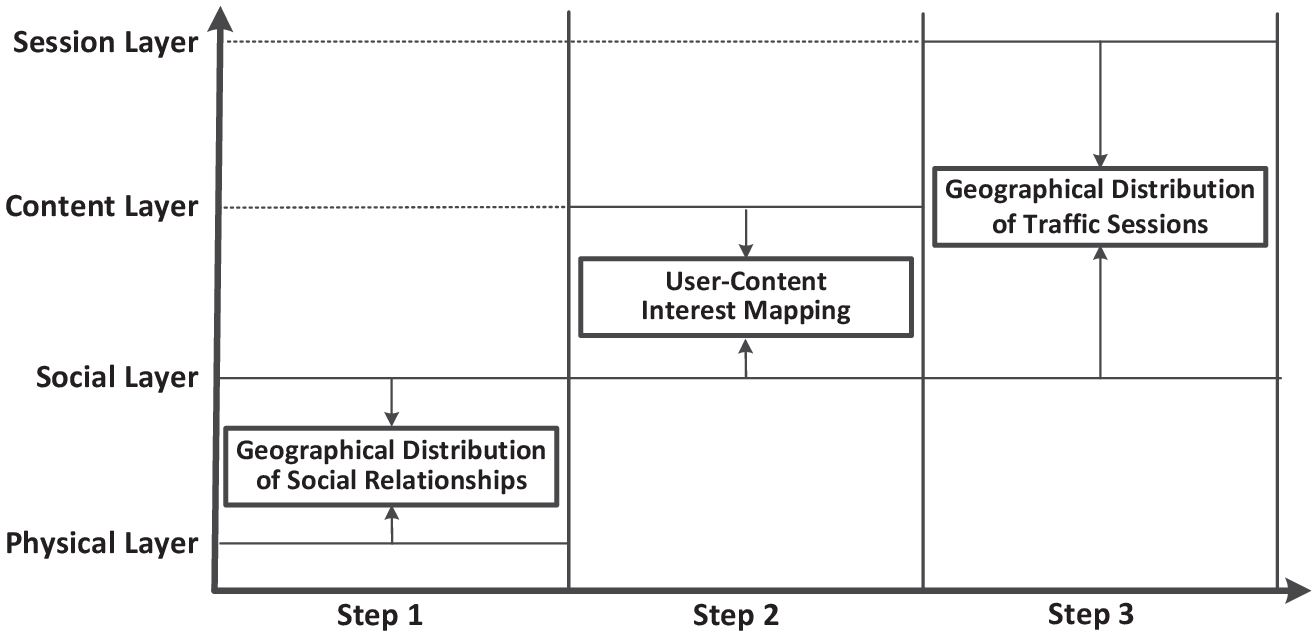}}
\caption{Three Steps for Modeling The Generation of Traffic Sessions. }\label{model-steps}
\end{center}
\vspace{-0.2in}
\end{figure}

Based on these three steps,  we can obtain the geographical distribution of traffic sessions, and thereby compute the bound on the aggregate transport distance over which the data is successfully transported from the source to the intended destinations.
Under a realistic assumption that every source sustains a data generating rate of constant order,
our result demonstrates that the transport complexity
for Social-InterestCast in an OSN over a carrier network with an optimal communication architecture
 varies in the range $\left[\Theta(n), \Theta(n^2)\right]$ depending on the clustering exponents of relationship degree, relationship formation and dissemination pattern, where $n$ is the user
number of the OSN.

The aforementioned results for Social-InterestCast clarify the differences from those for data dissemination in conventional communication networks,
and can serve as a valuable metric to measure network performance and the difficulty of data dissemination in large-scale networks.
Furthermore, if a specific carrier communication network is introduced, our results on transport distances can also play an important role in analyzing some system performances, e.g., network capacity and latency.

To the best of our knowledge, this is the first work to measure the transport difficulty for data dissemination in OSNs for modeling session patterns by taking into account the interest-driven characteristic and users' behaviour in real-world OSNs.

The rest of this paper is organized as follows.
Firstly, we show the related work in Section \ref{related-work}, and
give the metric of transport difficulty in Section \ref{sec-formulation}. In Section \ref{section-Network Model}, we propose our system model of Social-InterestCast.
In Section \ref{section-social-dissemination-systeme-setting}, we derive the transport complexity for the Social-InterestCast.
Finally, we draw a conclusion and make a discussion on our future work in Section \ref{sec-conclud-future}.

\section{Related Work}\label{related-work}
Online social networks (OSNs) provide a platform for hundreds of millions of the Internet users worldwide to produce and consume content. Users in OSNs have the access to the unprecedented large-scale information repository \cite{guille2013information}.
Moreover, OSNs play an important role in the information diffusion by increasing the spread of novel information and diverse viewpoints, and have shown their power in many situations \cite{bakshy2012role}.
There are some representative topics that have been extensively studied in the research community of OSNs, such as detecting popular topics \cite{mathioudakis2010twittermonitor}, digging potential popularity of contents \cite{altman2013stochastic},
modeling information diffusion \cite{myers2012information,diffusion2013}, identifying influential spreaders, leaders or followers  \cite{identi2013leader}, presenting influence mechanisms \cite{li2012modeling}, maximizing the spread of an information epidemic \cite{kandhwaycampaigning}, predicting the properties/signs of links \cite{ye2013predicting} and the missing preference of a user \cite{liu2013soco}, estimating the proximity of social networks \cite{song2012clustered},
and exploring security issues \cite{social2012com}, and so on.
However, most existing work mainly focused on the information diffusion scheme in overlay relationship networks of users in social networking sites/services (SNSs), \cite{Bayesian2012}.

Meanwhile, as SNSs become increasingly popular for information exchange, the traffic generated by social applications rapidly expands \cite{index2014global}.
A report of Shareaholic \cite{shareaholic2013} showed that,
between  November 2012 and November 2013, social media referral traffic from the top five social media sites increased by $111\%$ while search traffic from the top five search engines had decreased by $6\%$.
Therefore, besides the analysis of information diffusion schemes in overlay social networks \cite{bakshy2012role,mathioudakis2010twittermonitor,myers2012information,kandhwaycampaigning},
and the gain of social relationships in terms of the information dissemination \cite{yagan2013conjoining,chen2013social},
an in-depth understanding of the impact of increasing traffic generated by OSNs on carrier communication networks, e.g., the Internet, is convincingly necessary for evaluating current OSNs systems, optimizing network architectures and the deployment of servers for OSNs, and even designing future OSNs \cite{chen2013technological}.
To address this issue, we need to propose practical modeling and effective analytic methods for content distribution of OSNs implemented in carrier communication networks, since OSNs change both information propagation schemes and traffic session patterns in communication networks due to the involvement of overlay social relationships, users' preferences and decisions, \cite{evolution2014}.
Accordingly, in this paper, we aim at modeling content distribution \cite{amble2011content} in OSNs, and measuring its transport difficulty imposed on the carrier communication networks of OSNs.

The most relevant work investigating the load imposed on the carrier communication network is \cite{mobihoc2014},
where the metric called \emph{transport load}  (or \emph{traffic load} )  was proposed to quantify such a load.
In \cite{mobihoc2014},  only a theoretical  lower bound on transport load was derived without sufficiently clarifying the role of this result in measuring the transport difficulty  of data dissemination for specific online social networking applications.
We state that transport difficulty for a specific application should be an intrinsic property of this application when the carrier network is  given.
Consequently, the metric for transport difficulty should be a fundamental limit of a certain metric on transport burden.
From such a perspective,
we define the fundamental limit of  \emph{transport load} as  a new metric called \emph{transport complexity}, i.e., the \emph{minimum required} transport load for an OSN over a given carrier network.
Besides this,
comparing with \cite{mobihoc2014}, we make some significant improvements in this work:
In \cite{mobihoc2014}, the authors proposed a three-layered model to formulate data dissemination sessions for social applications in OSNs.
The session generation is simply modeled as Social-BroadCast, where the source broadcasts messages to all of its followers, and all followers have to be the passive destinations.
Apparently, they neglected the important features of session generation in OSNs, i.e., the fact that the generation of traffic sessions depends on both users' social relationships and the user-content interest mappings.
Accordingly, in this work,  we focus on an interest-driven session pattern, called Social-InteretCast.
To model its formation,  we introduce a \emph{content layer} as an interest-based filter to build interest links from users to messages, and propose a \emph{four-layered system model} as illustrated in Figure \ref{fig-4-layer-archi}.

\section{Metric of Transport Difficulty}\label{sec-formulation}

Considering an online social network (OSN), denoted by $\mathbb{N}$,  consisting of $n$ users, we denote the set of all users by $\mathcal{U}=\left\{u_i \right\}^{n}_{i=1}$. Let a subset
$\mathcal{S}=\left\{u_{\mathcal{S}, k}\right\} ^{n_s}_{k=1}\subseteq \mathcal{U}$ denote the set of all sources, where $|\mathcal{S}|=n_s$.
Denote a data dissemination session from a source $u_{\mathcal{S}, k}$ by an ordered pair $\mathbb{D}_{\mathcal{S}, k}=<u_{\mathcal{S}, k}, \mathcal{D}_{\mathcal{S}, k}>$, where
$\mathcal{D}_{\mathcal{S}, k}$ is the set of all destinations of $u_{\mathcal{S}, k}$.


In this work, we intend to investigate the transport difficulty for data dissemination in OSNs,
i.e.,  the load  on the carrier communication networks generated by a specific social applications.
To quantify such a load,  we  introduce a metric from \cite{mobihoc2014}, called \emph{transport load}, which
depends on two factors:  \emph{data requested rate} and \emph{data transport distance}.

\textbf{Data Requested Rate:}
Data requested rate is determined by QoS (quality of service) of
the application.  For data dissemination applications in OSNs, the QoS of data transport applications  is usually set up according to  \emph{the generating
rate of content at source users}, i.e., the so-called \emph{data arrival rate}.
Moreover, data requested rate is generally defined as a certain portion of data arrival rate.  In other words,
a portion of data arriving at the source user is requested to be successfully disseminated.
Therefore, we reasonably assume that \emph{the data requested rate has the same order as the data arrival rate}.

The temporal behavior of messages arriving at a user in an OSN has been addressed by analyzing  some real-life OSNs,\cite{milcom2010twitter,benevenuto2009characterizing}.
For example, Perera et al. \cite{milcom2010twitter}
developed a software architecture that uses a Twitter application program interface (API) to collect the tweets sent to specific users.
They indicated that the arrival process of new tweets to a user can be modeled as a Poisson Process.
In this paper, we just take it as an empirical argument for assuming  the data arrival for a user as a data source follows a Poisson Process, \cite{milcom2010twitter}.
 In our work, for each session $\mathbb{D}_{\mathcal{S}, k}=<u_{\mathcal{S}, k}, \mathcal{D}_{\mathcal{S}, k}>$,
 we simply set the \emph{data requested rate} to be a portion of the data arrival rate.
Then, we can denote the data requested rate by
a vector
$\mathbf{\Lambda}_{\mathcal{S}}=(\lambda_{\mathcal{S},1}, \lambda_{\mathcal{S},2}, \cdots, \lambda_{\mathcal{S},n_s}),$
 where $\lambda_{\mathcal{S},k}$ is the rate of a Poisson Process at user $u_{\mathcal{S}, k}$ (for $k=1,2,\cdots, n_s$).

In practice, the data arrival rate is dispensable on the scale of the specific OSN, i.e., the value of $n$, although the data arrival rate depends on many factors, such as the specific form and quality of social services.
Combining the facts that the \emph{data requested rate} is a certain portion of the data arrival rate,
we can make a reasonable and practical assumption that $\lambda_{\mathcal{S},k}=\Theta(1)$ for $k=1,2,\cdots, n_s$.

In this work, we aim to analyze the fundamental limits on the transport load for data dissemination in OSNs according to the network size.
Therefore, it is appropriate to note at this point that the specific distribution of data requested rate has no impact on the results (in order sense) as long as it holds that $\lambda_{\mathcal{S},k}=\Theta(1)$ for $k=1,2,\cdots, n_s$.
This is why we do not make an intensive study of the specific distribution of data requested rates.


\textbf{Data Transport Distance:}
The data transport distance  is comprehensively determined by traffic session pattern of the application,
communication network architecture, and transmission schemes.
For a given transmission scheme in a given  OSN $\mathbb{N}$, say $\mathbf{S}_\mathbb{N}$, define a vector
\[\mathbf{D}_{\mathcal{S}}(\mathbf{S}_\mathbb{N})=\left(d_{\mathcal{S},1}(\mathbf{S}_\mathbb{N}), d_{\mathcal{S},2}(\mathbf{S}_\mathbb{N}), \cdots, d_{\mathcal{S},n_s}(\mathbf{S}_\mathbb{N}) \right),\]
 where
$d_{\mathcal{S}, k}(\mathbf{S}_\mathbb{N})$ represents the \emph{transport distance} over which the message for session $\mathbb{D}_{\mathcal{S}, k}$ is successfully transported from the source $u_{\mathcal{S}, k}$ to all destinations.


\textbf{Transport Load and Transport Complexity:}
In the OSN $\mathbb{N}$, given a specific carrier communication network, define the transport load for a dissemination session, say $\mathbb{D}_{\mathcal{S}, k}$, as
\[
\widetilde{\mathbf{L}}_{\mathbb{N},\mathcal{S}}(\mathbb{D}_{\mathcal{S}, k})=\lambda_{\mathcal{S},k} \cdot d_{\mathcal{S},k}(\mathbf{S}_\mathbb{N}).\] Furthermore, the aggregated transport load for dissemination sessions from all sources in  $\mathcal{S}$ can be  defined as
\begin{equation}\label{equ-transport-load}
\mathbf{L}_{\mathbb{N},\mathcal{S}}= \min\nolimits_{\mathbf{S}_\mathbb{N}\in \mathbb{S}} \mathbf{\Lambda}_{\mathcal{S}}  \ast \mathbf{D}_{\mathcal{S}}(\mathbf{S}_\mathbb{N}),
\end{equation}
where $\mathbb{S}$ is the set of all feasible transmission schemes, and $\ast$ is an inner product.

Based on the definition of transport load, we further introduce the \emph{feasible transport load}.

\begin{defn}[Feasible Transport Load]\label{feasible-traffic-load}
For a social data dissemination with a set of social sessions $\mathbb{D}_{\mathcal{S}}=\{\mathbb{D}_{\mathcal{S}, k}\}_{k=1}^{n}$, we say that the  transport load
$\mathbf{L}_{\mathbb{N},\mathcal{S}}$
is \emph{feasible} if and only if
there exists an appropriate transmission scheme with a communication deployment, denoted  by $\mathbf{S}_\mathbb{N}$,
such that
it holds that $\mathbf{D}_{\mathcal{S}}(\mathbf{S}_\mathbb{N}) * \mathbf{\Lambda}_{\mathcal{S}}  \leq \mathbf{L}_{\mathbb{N},\mathcal{S}}$ ensuring that the network throughput of $\mathbf{\Lambda}_{\mathcal{S}} $ is achievable.
\end{defn}

Based on the definition of feasible transport load, we finally define the \emph{transport complexity}.
\begin{defn}[Transport Complexity]\label{tansmission-complexity}
We say that the \emph{transport complexity} of the class of random social data disseminations $\mathbb{D}_{\mathcal{S}}$ is
of order $ \Theta(f(n)) $ bit-meters per second
if there are deterministic constants $ c_1 > 0 $ and $ c_2 < \infty $ such that:
There exists a communication architecture $\mathbb{N}$ and corresponding transmission schemes such that
\begin{equation*}
  \lim_{n\longrightarrow\infty} \Pr (\mathbf{L}_{\mathbb{N},\mathcal{S}} = c_1 \cdot f(n) \mbox{ is feasible}) = 1,
\end{equation*}
and  for any possible communication architectures and transmission schemes, it holds that:
\begin{equation*}
   \liminf_{n\longrightarrow\infty} \Pr(\mathbf{L}_{\mathbb{N},\mathcal{S}} = c_2  \cdot f(n) \mbox{ is feasible}) <1.
\end{equation*}
\end{defn}

\section{Model of Social-InterestCast}\label{section-Network Model}

For each session, the geographical distribution of the source and destination(s) plays a key role in generating the transport load.
So it is critical to analyze  the correlation between the spatial distribution of sessions and geographical distribution of users in online social networks (OSNs).

To address this issue, we propose a four-layered model,  consisting of
the  physical layer,  social layer, content layer, and  session layer, as illustrated in Figure \ref{fig-4-layer-archi}.

\subsection{Four-Layered System Model}

\subsubsection{Physical Layer- Physical Network Deployment}
The deployment of the so-called physical network can be divided into two parts. The first is the geographical distribution of  social users. The second is the communication architecture of  the carrier network.

\textbf{Geographical Distribution of Social Users:} We consider the random network consisting of a random number $N$ (with  $\mathbf{E}(N)=n$)\footnote{Throughout the paper,  let $\mathbf{E}[X]$ denote the mean of a random variable $X$.} users who are randomly distributed
over a square region of area $S:=n$, where $\mathbf{E}[N]=n$.
To avoid border effects, we consider wraparound
conditions at the network edges, i.e., the network area
is assumed to be the surface of a two-dimensional Torus $\mathcal{O}$.
To simplify the
description, we assume that the number of nodes is exactly $n$, and denote the set of nodes by $\mathcal{V}=\{v_k\}_{k=1}^n$,
without changing our results in order sense.
We make a compromise in the generality and practicality of the geographical distribution of social users,
in order to concentrate on clarifying the impacts of users' interest on the session formation.
Specifically, we follow the setting where all users are distributed according to a homogeneous Poisson point process, taking no account of the inhomogeneous property of the uneven population distribution in real-life OSNs.
The derived results are expected to serve as the basis for investigating more realistic scenarios under
the more practical but complex deployment models, such as the Clustering Random Model (CRM) according to the shotnoise Cox process \cite{alfano2009capacity} and Multi-center Gaussian Model
(MGM) in \cite{mgm2012}.

\textbf{Communication Architecture of  Carrier Network:}
For online social networking services, we state that the carrier network is indeed the mobile Internet.
This means that there is hardly a uniform communication architecture, e.g., centralized or distributed network architecture, practically characterizing the architecture of a real-world carrier network for OSNs.
This is also the reason why we derive the transport complexity
for Social-InterestCast sessions  over the carrier network with optimal communication architecture.
The results can serve as the common lower bounds on transport complexity for Social-InterestCast sessions over any carrier networks.

\subsubsection{Social Layer- Social Relationship}
To model the geographical distribution of relationships, we introduce the
\emph{population-distance-based model}  from
\cite{mobihoc2014}.  
In \cite{mobihoc2014}, Wang et al. provided a numerical evaluation for this social model based on a Brightkite users' dataset \cite{Brightkite2007}.
Before deciding to adopt this model, we complement an evaluation based on another dataset from \cite{Hossmann2011},
i.e.,  a Gowalla users' dataset.
For the completeness, we include the evaluations based on these datasets in Appendix \ref{appendix-sec-evaluation}.

%
Let $\mathcal{D}(u, r)$ denote the disk centered at a node $u$ with radius $r$ in the deployment region $\mathcal{O}$, and let
$N(u, r)$ denote the number of nodes contained in  $\mathcal{D}(u, r)$.
Then, for any two nodes, say $u$ and $v$, we can define the \emph{population-distance} from $u$ to $v$ as $N(u, |u-v|)$,
where $|u-v|$ denotes the Euclidean distance between node $u$ and node $v$.

For completeness,  we include the population-distance-based model as follows:
For a node $v_k\in \mathcal{V}$, we construct its relationship set $\mathcal{F}_k$ of $q_k$ ($q_k\geq 1$) follower nodes by the following procedures: For a particular node $v_k\in \mathcal{V}$, denote the number of followers by $q_k$, we assume that it follows a Zipf's
distribution \cite{manning1999foundations}, i.e.,
\begin{equation}\label{equ-zip-destinations}
\Pr(q_k=l)={\left(\sum\nolimits_{j=1}^{n-1}j^{-\gamma}\right)^{-1}} \cdot {l^{-\gamma}}.
\end{equation}
Note that we accordingly give a numerical evaluation based on Gowalla dataset \cite{Hossmann2011} for the Zipf's degree distribution in Appendix \ref{appendix-ddobu234324}.

Then, to model the geographical distribution of social relationships, we make the position of node $v_k$ as the \emph{reference point} and choose $q_k$ \emph{points} independently on the torus region $\mathcal{O}$ according to a  probability distribution with the following density function:
\begin{equation}\label{equ-social-distribution}
f_{v_k}(X) = \Phi_k(S, \beta) \cdot \left ( \mathbf{E}[N(v_k, |X-v_k|)] +1 \right )^{-\beta} ,
\end{equation}
where the random variable $X:=(x,y)$ denotes the position of a selected point in the deployment region,
and $\beta\in [0, \infty)$ represents the clustering exponent of relationship formation; the coefficient $\Phi_k(S, \beta)>0$  depends on $\beta$ and $S$ (the area of deployment region), and it satisfies:
\begin{equation}\label{equ-formation-completeness-condition}
\Phi_k(S, \beta) \cdot \int_{\mathcal{O}} \left ( \mathbf{E}[N(v_k, |X-v_k|)] +1 \right )^{-\beta} d X =1.
\end{equation}

Further, we determine the nearest followers for specific users.
Let $\mathcal{A}_k=\{p_{k_i}\}_{i=1}^{q_k}$ denote the set of these $q_k$ points.
Let $v_{k_i}$ be
the nearest node to  $p_{k_i}$, for $1 \leq i \leq q_k$ (ties are broken randomly). Denote the set of these $q_k$ nodes by
$\mathcal{F}_k=\{v_{k_i}\}_{i=1}^{q_k}$.
We call point $p_{k_i}$ the \emph{anchor point} of $v_{k_i}$, and define a set
$\mathcal{P}_k:= \{v_k\} \cup \mathcal{A}_k.$

In this model, we can observe that the degree distribution above depends on the specific network size, i.e., the number of users $n$.
Note that the correlation between the users' degree distribution and geographical distribution should not be neglected, although we simplify it in this work. We will entirely address this issue in the future work. Throughout this paper, we use $\mathbb{P}(\gamma, \beta)$ to denote the population-distance-based model. Here, $\gamma$ and
$\beta$ denote the clustering exponents of relationship degree and relationship formation, respectively.
For their detailed meanings, please refer to Eq.(\ref{equ-zip-destinations}) and Eq.(\ref{equ-social-distribution}), respectively.

\subsubsection{Content Layer- User-Content Interest Mapping}
In this work, we focus on the case where a user only view the messages within his interest.
Then, modeling the mapping between users and content is the first critical step for the final generation of traffic sessions.
The basic idea for this issue is to build the interest links from users to messages according to the similarity
of users' interests to message semantics.
Specifically, we dig the user's interest by abstracting the topics of all his posted messages.

\subsubsection{Session Layer- Data Traffic Session}\label{subsec-social-session-layer}
When observing users' behaviour in real-life OSNs, there exists a significant feature of traffic session in OSNs, i.e.,
the users' intrinsic interest and subjective choice  are often indispensable for forming the sessions, \cite{li2014understanding}.
For all the common traffic patterns in OSN, we focus on a typical data dissemination, called SocialCast, which can be divided into two successive phases: \emph{Passive Phase} and \emph{Initiative Phase}.
In the passive phase, a source user holding a message sends self-appointedly the digest of this message to all his followers. The followers passively receive this digest.
This process just acts like a broadcast in all the followers of this source user.
This process  is called \emph{Social-BroadCast} in \cite{mobihoc2014}.
In the initiative phase, according to the interest, a user will autonomously determine whether to download the complete message based on the digest of a message from who he follows.
We define such an interest-driven dissemination process \emph{Social-InterestCast}.
The Social-InterestCast is defined by taking into account specific user interests and their effects on session generation.
Our focus in this work is Social-InterestCast.

\subsection{Social-InterestCast}\label{social_interestcatsmodel}
In real-life OSNs,  there exists naturally a common phenomenon: When a user broadcasts messages to others, the number of potential destinations is depending on underlying relationships among users (both the source and destinations), \cite{li2014understanding}.
For a delivered message, to explore the characteristics of source user and final destinations, we begin with studying relationships between users and messages. Then, it is common that if a user broadcasts a message characterized by several words (abstract/digest) in his Facebook, only the users whose interests are consistent with the topic of the message can be the final destinations. We say that this behavior occurs due to an underlying \emph{user-message mapping} between user set and message set. For such type of  session, we call it \emph{Social-InterestCast}, which is different from Unicast in \cite{azimdoost2012globecomcapacity} and Social-BroadCast in \cite{mobihoc2014}. In the following,
we firstly dig the distribution of destinations of Social-InterestCast sessions based on a real-life dataset of Foursquare \cite{Bao2012Location}.

\subsubsection{Analyzing Social-InterestCast Formation}\label{social_interestcatsmodel-subusermbsub}
In this part, we aim to analyze the distribution of destination number for Social-InterestCast by exploring the session formation mechanism
over a real-life Foursquare dataset \cite{Bao2012Location}.
Foursquare was created in 2009. It is a location-based social networking service provider where users share their locations via ``checking-in" function. The dataset contains $104,478$ check-in tips generated by $31,544$ users in Los Angeles (LA).
It provides the follower lists and check-in messages of all users.


In real-life OSNs, the users' \emph{intrinsic interest} and \emph{subjective choice} are usually indispensable for forming the traffic sessions.
Under such a common type of sessions, a user will autonomously determine whether to download the content from whom he follows depending on his interest, \cite{borghol2012untold}.
We apply the \emph{Latent Dirichlet Allocation} (LDA) \cite{Blei2003Latent} to extract topics from the check-in messages generated by users.  LDA assumes that words of each document are drawn from a mixture of topics.
We define all messages generated by a user, say $v$, as a \emph{user document} $c_{v}$.
Then, we apply LDA model to extract the topic distribution $\theta_{c_{v}}$, serving as user's interests. Figure \ref{fig-result-lda-model} depicts the LDA model.
With the user's interest distribution $\theta_{c_{v}}$, we introduce \emph{Jensen-Shannon divergence} ${JSD}(P||Q)$ \cite{dhillon2002enhanced} to measure the similarity between users' interest $\theta_{c_{v}}$ and newly arrived message $\theta_{c_{new}}$.

\begin{figure}[t]
\begin{center}
\scalebox{0.92}{\includegraphics{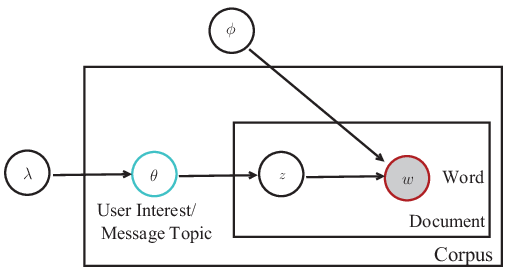}}
\caption{$\lambda$  is the hyperparameter of LDA model, $\phi$ is the topic-word distribution, $\theta$ is the topic distribution of user document, i.e., user interest, $z$ denotes the latent topic in document $c_{v}$, and $w$ is an observed word.} \label{fig-result-lda-model}
\end{center}
\vspace{-0.2in}
\end{figure}

In our evaluation, we assume that only when the user's interest $\theta_{c_{v}}$ is similar to the message's topic distribution $\theta_{c_{new}}$, user $v$ will further view the newly arrived message $c_{new}$.
In other words,  when the value of ${JSD}(\theta_{c_{v}}||\theta_{c_{new}})$ is smaller than that of a pre-defined threshold $\delta$, the user $v$ is selected as a final destination of the message $c_{new}$.

\begin{figure}[t]
\begin{center}
\scalebox{0.98}{\includegraphics{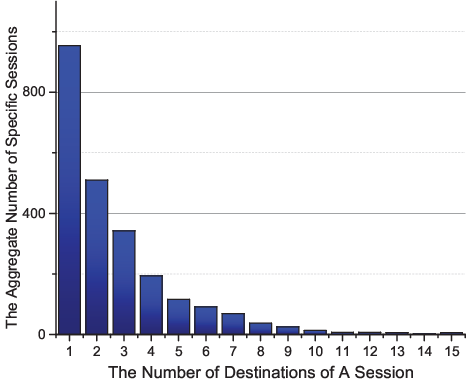}}
\caption{Destinations Distribution with Relationship Degree $q_k=22$.}\label{fig-destination-distribution-22}
\end{center}
\vspace{-0.2in}
\end{figure}
As illustrated in Figure \ref{fig-destination-distribution-22},
one of the experimental results based on Foursquare dataset \cite{Bao2012Location}
 shows the distribution of session destinations initiated from users who have $22$ followers.
In Figure \ref{fig-destination-distribution-22},
the $X$-axis denotes the number of destinations at a certain degree $22$, the $Y$-axis denotes the number of cases where destination number equals  $X$.  We can see that $Y$ decreases rapidly at first and then gently with the increasing of $X$.

To provide more insights, Figure \ref{fig-destination-distribution} illustrates the distributions of destination number for different relationship degree $q_k$ with the threshold $\delta = 6.5$.
Since users with larger relationship degree cannot show convincingly the statistical characteristics due to the small number of such users (the small sample size), we have removed those users whose degree is larger
than $32$ to obtain a more representative result.
 In Figure \ref{fig-destination-distribution}, the experimental result in each subfigure shows that the distribution of destination number for each Social-InterestCast session appears a long-tailed feature.

\begin{figure}[t]
\begin{center}
\scalebox{0.5}{\includegraphics{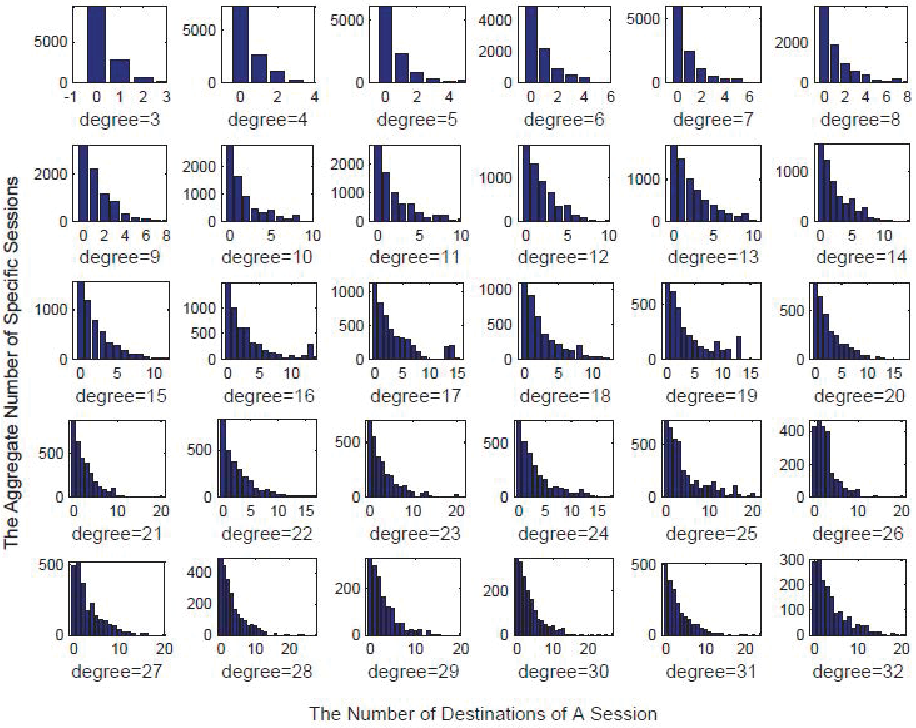}}
\caption{Destinations Distribution in Foursquare.}\label{fig-destination-distribution}
\end{center}
\vspace{-0.2in}
\end{figure}

Note that the long-tailed property in the distribution of destination number is just derived under our hypothetical formation mechanism of Social-InterestCast using LDA model.
As an important empirical argument,  Li et al. \cite{li2014understanding} has provided a numerical validation based on RenRen dataset for
such a long-tailed distribution by counting the ratio of the number of viewed videos and that of the received videos from friends.

\subsubsection{Modeling Social-InterestCast}\label{model-for-interestcast}
Based on the result in Section \ref{social_interestcatsmodel-subusermbsub},  we assume that the distribution of session destinations for Social-InterestCast follows approximately a Zipf's distribution. Specifically, for a source node $v_{k}$, based on the dependency among followers' decisions and the attractivity difference of message content, we make a reasonable assumption that the number of destinations follows a Zipf's distribution whose parameters depend on the relationship degree of $v_{k}$, \ie,
\begin{equation}\label{zipf-final-destination}
  \mathrm{Pr}(d_{k}=d|q_{k}=l)={\left(\sum\nolimits_{m=1}^{l}m^{-\varphi_{k}}\right)^{-1}} \cdot  d^{-\varphi_{k}},
\end{equation}
where $ d_{k} $ is the number of  final destinations of a session initiated by $v_{k}$, and $ \varphi_k\in[0,\infty) $ is the exponent of data dissemination. Here, for problem simplification, we first study the special case where
$\varphi_{k} \equiv  \varphi $ for every $v_k$.


In the following Section \ref{section-social-dissemination-systeme-setting}, we mainly study this type of data disseminations and give the corresponding results for transport complexity.

\section{Transport Complexity for \\ Social-InterestCast}\label{section-social-dissemination-systeme-setting}
In this section, we aim to derive the transport complexity for Social-InterestCast.

\subsection{Social-InterestCast Sessions}\label{interestcast-session}
Denote a Social-InterestCast session by an ordered pair $\mathbb{D}^{\mathrm{I}}_k=<v_{k}, \mathcal{I}_k>$, where $ v_{k} $ is the source and each element $ v_{k_{i}} $ in $\mathcal{I}_k=\{v_{k_i}\}_{i=1}^{d_k}$ is the nearest node to the corresponding $ p_{k_{i}} $ in $ \mathcal{A}^\mathrm{I}_k=\{p_{k_i}\}_{i=1}^{d_k} $, the random variable $d_{k}$ denotes the number of potential destinations for session $\mathbb{D}^{\mathrm{I}}_k$, i.e., the followers of $v_{k}$ who are interested in the message from $v_{k}$ in this session.
We call point $p_{k_{i}}$ the \emph{anchor point} of $v_{k_{i}}$, and define a set $\mathcal{P}^\mathrm{I}_{k} := \{v_{k}\} \cup \mathcal{A}^\mathrm{I}_k$. Then, we can get the following lemma.

\begin{lem}\label{lem-EMST-B-D} For a Social-InterestCast session $\mathbb{D}^{\mathrm{I}}_k$, when $ d_{k}=\omega(1) $, with probability $1$, it holds that
\begin{equation*}
|\EMST(\mathcal{A}^\mathrm{I}_k)|=\Theta(|\EMST(\mathcal{P}^\mathrm{I}_{k})|)=\Theta(L^\mathrm{I}_{\mathcal{P}}(\beta, d_k)),
\end{equation*}
where $\EMST(\cdot)$ denotes the Euclidean minimum spanning tree over a set, and
\begin{equation}\label{eq-L-P-beta-9765}
L^\mathrm{I}_{\mathcal{P}}(\beta, d_k) = \left\{
  \begin{array}{ll}
\Theta\left( \sqrt{d_k} \right)  , &  \beta>2;  \\
\Theta\left(\sqrt{d_k} \cdot \log n  \right), &  \beta=2;  \\
\Theta\left(\sqrt{d_k} \cdot n^{1-\frac{\beta}{2}} \right)  , & 1< \beta < 2;  \\
 \Theta\left(\sqrt{d_k} \cdot \sqrt{\frac{n}{\log n}} \right) , &  \beta=1;  \\
  \Theta\left(\sqrt{d_k} \cdot \sqrt{n} \right)  , &  0\leq\beta <1.
  \end{array}
\right.
\end{equation}
\end{lem}
\begin{proof}
By Lemma \ref{lem-growthsteel} in Appendix \ref{apeneidisx-ausouful-lmemas}, with probability $1$, it holds that
\[|\EMST(\mathcal{A}^\mathrm{I}_k)|=\Theta(\sqrt{d_k} \cdot \int_{\mathcal{O}}\sqrt{f_{X_0}(X)} dX),\]
where
\[
f_{X_0}(X)=\left\{
  \begin{array}{ll}
\Theta\left( (|X-X_0|^2 +1)^{-\beta} \right)  , &  \beta>1;  \\
\Theta\left(\left(\log n \cdot (|X-X_0|^2 +1) \right )^{-1} \right)    , &  \beta=1;  \\
\Theta\left( n^{\beta-1} \cdot   (|X-X_0|^2 +1)^{-\beta}  \right)   , &  0\leq\beta <1.
  \end{array}
\right.
\]
Next, we compute $\int_{\mathcal{O}}\sqrt{f_{X_0}(X)} dX$.
According to the value of $ \beta $, we have:

(1) When $ 0 \leq\beta <1 $,
\begin{equation*}
  \int_{\mathcal{O}}\sqrt{f_{X_0}(X)} dX =  \Theta\left(n^{\frac{\beta-1}{2}} \cdot \int_{\mathcal{O}}\frac{dX}{\left(|X-X_{0}|^{2}+1\right)^{\frac{\beta}{2}}}\right)
\end{equation*}
$~~~~~~~~~~~~~~~~~~~~~~~= \Theta\left(n^{\frac{\beta-1}{2}} \cdot n^{1-\frac{\beta}{2}}\right) = \Theta\left(\sqrt{n}\right) $.

(2) When $ \beta=1 $,
\begin{equation*}
  \int_{\mathcal{O}}\sqrt{f_{X_0}(X)}dX = \Theta\left(\frac{1}{\sqrt{\log n}}\cdot \int_{\mathcal{O}}\frac{dX}{\left(|X-X_{0}|^{2}+1\right)^{\frac{1}{2}}}\right)
\end{equation*}
$~~~~~~~~~~~~~~~~~~~~~~~= \Theta\left(\frac{1}{\sqrt{\log n}} \cdot \sqrt{n}\right) = \Theta\left(\sqrt{\frac{n}{\log n}}\right) $.

(3) When $ \beta>1 $,
\begin{equation*}
\int_{\mathcal{O}}\sqrt{f_{X_0}(X)}dX = \Theta\left( \int_{\mathcal{O}}\frac{dX}{\left(|X-X_{0}|^{2}+1\right)^{\frac{\beta}{2}}}\right).
\end{equation*}

Especially,

when $ 1<\beta<2 $, $ \Theta\left( \int_{\mathcal{O}}\frac{dX}{\left(|X-X_{0}|^{2}+1\right)^{\frac{\beta}{2}}}\right) = \Theta\left(n^{1-\frac{\beta}{2}}\right) $;

when $ \beta=2 $, $ \Theta\left( \int_{\mathcal{O}}\frac{dX}{\left(|X-X_{0}|^{2}+1\right)}\right) $=$ \Theta\left(\log n\right) $;

when $ \beta>2 $, $ \Theta\left( \int_{\mathcal{O}}\frac{dX}{\left(|X-X_{0}|^{2}+1\right)^{\frac{\beta}{2}}}\right) $=$ \Theta\left(1\right) $.

Then, by summarizing the derived results above, we can obtain that
\[
|\EMST(\mathcal{A}^\mathrm{I}_k)|=\Theta(L^\mathrm{I}_{\mathcal{P}}(\beta, d_k)),
\]
where $L^\mathrm{I}_{\mathcal{P}}(\beta, d_k)$ is defined in Eq.(\ref{eq-L-P-beta-9765}).

Let $\underline{L}$ denote
the smallest distance from the points in $\mathcal{A}^\mathrm{I}_k$ to point $X_0$. We can get that
\[|\EMST(\mathcal{A}^\mathrm{I}_k)|   \leq       |\EMST(\mathcal{P}^\mathrm{I}_k)|    \leq        |\EMST(\mathcal{A}^\mathrm{I}_k)|+\underline{L},\]
Furthermore, according to the fact that $\underline{L}=O(|\EMST(\mathcal{A}^\mathrm{I}_k)|)$  with probability $1$,  we get
$|\EMST(\mathcal{P}^\mathrm{I}_k)|=\Theta(|\EMST(\mathcal{A}^\mathrm{I}_k)|)$,
which finally completes the proof.
\end{proof}

\subsection{Main Results on Transport Complexity}\label{Z-D-result-interestcast}
\subsubsection{Bounds on Aggregated Transport Distance}
The bound on the transport complexity depends on the value of $ \sum\nolimits_{k=1}^{n}|\mathrm{EMST}(\mathcal{P}^\mathrm{I}_k)| $, which we compute in the following theorem.
\begin{thm}\label{lem-delta-0-all-est} For all Social-InterestCast sessions $ \{\mathbb{D}^{\mathrm{I}}_k\}_{k=1}^{n} $ with the Zipf's distribution as defined in Eq.(\ref{zipf-final-destination}), with high probability, the order of $\sum\nolimits_{k=1}^{n}|\mathrm{EMST}(\mathcal{P}^\mathrm{I}_k)| $ holds as presented in Table \ref{tab-lower-bound-EMST}.
\end{thm}

\begin{table*}[!t]
\renewcommand{\arraystretch}{1.0}
\centering
\caption{The Order of $ \sum\nolimits_{k=1}^{n}|\mathrm{EMST}(\mathcal{P}^\mathrm{I}_k)|$}
\label{tab-lower-bound-EMST}
\scalebox{0.75}
{
\begin{tabular}{|p{1.8cm}||p{3.85cm}|p{3.85cm}|p{3.85cm}|p{3.85cm}|p{3.85cm}|}
\hline
$\varphi$ $\backslash$ $\beta$ & $\beta>2$ & $\beta=2$ & $ 1<\beta<2 $ & $\beta=1$ & $0\leq\beta<1$ \\
\hline
\hline

$\varphi>\frac{3}{2}$ &
 $ \Theta(n)$ , $\gamma\geq0$; &
 $ \Theta(n \cdot \log n)$ , $\gamma\geq0$; &
 $ \Theta(n^{2-\frac{\beta}{2}})$ , $\gamma\geq0$; &
 $ \Theta(n^{\frac{3}{2}} / \sqrt{\log n})$ , $\gamma\geq0$; &
 $ \Theta(n^{\frac{3}{2}})$ , $\gamma\geq0$. \\
\hline

$\varphi=\frac{3}{2}$ & $\left\{
  \begin{array}{ll}
\Theta(n)  ,   \\
\qquad \gamma>1; \\
\Theta(n \cdot \log n)  ,  \\
 \qquad 0\leq\gamma\leq1.\\
  \end{array}
\right. $ &
$\left\{
  \begin{array}{ll}
\Theta(n \cdot \log n)  ,   \\
\qquad \gamma>1; \\
\Theta(n \cdot (\log n)^2)  ,  \\
 \qquad 0\leq\gamma\leq1.\\
  \end{array}
\right. $ &
$\left\{
  \begin{array}{ll}
\Theta(n^{2-\frac{\beta}{2}})  ,   \\
\qquad \gamma>1; \\
\Theta(n^{2-\frac{\beta}{2}} \cdot \log n)  ,  \\
 \qquad 0\leq\gamma\leq1.\\
  \end{array}
\right. $ &
$\left\{
  \begin{array}{ll}
\Theta(n^{\frac{3}{2}} / \sqrt{\log n})  ,   \\
\qquad \gamma>1; \\
\Theta(n^{\frac{3}{2}} \cdot \sqrt{\log n})  ,  \\
\qquad  0\leq\gamma\leq1.\\
  \end{array}
\right. $ &
$\left\{
  \begin{array}{ll}
\Theta(n^{\frac{3}{2}})  ,   \\
\qquad \gamma>1; \\
\Theta(n^{\frac{3}{2}} \cdot \log n)  ,  \\
\qquad 0\leq\gamma\leq1.\\
  \end{array}
\right. $  \\
\hline

$1<\varphi<\frac{3}{2}$ & $\left\{
  \begin{array}{ll}
\Theta(n)  ,   \\
\qquad \gamma>\frac{5}{2}-\varphi;\\
\Theta(n \cdot \log n) , \\
\qquad \gamma=\frac{5}{2}-\varphi;  \\
\Theta(n^{\frac{7}{2}-\gamma-\varphi}) , \\
 \qquad 1<\gamma<\frac{5}{2}-\varphi;  \\
\Theta(n^{\frac{5}{2}-\varphi} / \log n) , \\
 \qquad \gamma=1;  \\
\Theta(n^{\frac{5}{2}-\varphi})  ,  \\
 \qquad 0\leq\gamma<1.\\
  \end{array}
\right. $ &
$\left\{
  \begin{array}{ll}
\Theta(n \cdot \log n)  ,   \\
\qquad \gamma>\frac{5}{2}-\varphi;\\
\Theta(n \cdot (\log n)^2) , \\
\qquad \gamma=\frac{5}{2}-\varphi;  \\
\Theta(n^{\frac{7}{2}-\gamma-\varphi} \cdot \log n) , \\
\qquad 1<\gamma<\frac{5}{2}-\varphi;  \\
\Theta(n^{\frac{5}{2}-\varphi}) , \\
 \qquad \gamma=1;  \\
\Theta(n^{\frac{5}{2}-\varphi} \cdot \log n)  ,  \\
\qquad 0\leq\gamma<1.\\
  \end{array}
\right. $ &
$\left\{
  \begin{array}{ll}
\Theta(n^{2-\frac{\beta}{2}})  ,   \\
\qquad \gamma>\frac{5}{2}-\varphi;\\
\Theta(n^{2-\frac{\beta}{2}} \cdot \log n) , \\
\qquad \gamma=\frac{5}{2}-\varphi;  \\
\Theta(n^{\frac{9}{2}-\gamma-\varphi-\frac{\beta}{2}}) , \\
\qquad 1<\gamma<\frac{5}{2}-\varphi;  \\
\Theta(n^{\frac{7}{2}-\varphi-\frac{\beta}{2}} / \log n) , \\
\qquad \gamma=1;  \\
\Theta(n^{\frac{7}{2}-\varphi-\frac{\beta}{2}})  ,  \\
\qquad 0\leq\gamma<1.\\
  \end{array}
\right. $ &
$\left\{
  \begin{array}{ll}
\Theta(n^{\frac{3}{2}} / \sqrt{\log n})  ,   \\
\qquad \gamma>\frac{5}{2}-\varphi;\\
\Theta(n^{\frac{3}{2}} \cdot \sqrt{\log n}) , \\
\qquad \gamma=\frac{5}{2}-\varphi;  \\
\Theta(n^{4-\gamma-\varphi} / \sqrt{\log n}) , \\
\qquad 1<\gamma<\frac{5}{2}-\varphi;  \\
\Theta(n^{3-\varphi} / (\log n)^{\frac{3}{2}}) , \\
\qquad \gamma=1;  \\
\Theta(n^{3-\varphi} / \sqrt{\log n})  ,  \\
\qquad 0\leq\gamma<1.\\
  \end{array}
\right. $ &
$\left\{
  \begin{array}{ll}
\Theta(n^{\frac{3}{2}})  ,   \\
\qquad \gamma>\frac{5}{2}-\varphi;\\
\Theta(n^{\frac{3}{2}} \cdot \log n) , \\
\qquad \gamma=\frac{5}{2}-\varphi;  \\
\Theta(n^{4-\gamma-\varphi}) , \\
\qquad 1<\gamma<\frac{5}{2}-\varphi;  \\
\Theta(n^{3-\varphi} / \log n) , \\
\qquad \gamma=1;  \\
\Theta(n^{3-\varphi})  ,  \\
\qquad 0\leq\gamma<1.\\
  \end{array}
\right. $ \\
\hline

$\varphi=1$ & $\left\{
  \begin{array}{ll}
\Theta(n)  , \\
\qquad  \gamma\geq\frac{3}{2}; \\
\Theta(n^{\frac{5}{2}-\gamma} / \log n) , \\
 \qquad 1<\gamma<\frac{3}{2};  \\
\Theta(n^{\frac{3}{2}} / (\log n)^2) , \\
 \qquad \gamma=1;  \\
\Theta(n^{\frac{3}{2}} / \log n)  , \\
\qquad  0\leq\gamma<1.\\
  \end{array}
\right. $ &
$\left\{
  \begin{array}{ll}
\Theta(n \cdot \log n)  , \\
\qquad  \gamma\geq\frac{3}{2}; \\
\Theta(n^{\frac{5}{2}-\gamma}) , \\
 \qquad 1<\gamma<\frac{3}{2};  \\
\Theta(n^{\frac{3}{2}} / \log n) , \\
 \qquad \gamma=1;  \\
\Theta(n^{\frac{3}{2}})  , \\
\qquad  0\leq\gamma<1.\\
  \end{array}
\right. $ &
$\left\{
  \begin{array}{ll}
\Theta(n^{2-\frac{\beta}{2}})  , \\
\qquad  \gamma\geq\frac{3}{2}; \\
\Theta(n^{\frac{7}{2}-\gamma-\frac{\beta}{2}} / \log n) , \\
\qquad  1<\gamma<\frac{3}{2};  \\
\Theta(n^{\frac{5}{2}-\frac{\beta}{2}} / (\log n)^2) , \\
 \qquad \gamma=1;  \\
\Theta(n^{\frac{5}{2}-\frac{\beta}{2}} / \log n)  , \\
\qquad  0\leq\gamma<1.\\
  \end{array}
\right. $ &
$\left\{
  \begin{array}{ll}
\Theta(n^{\frac{3}{2}} / \sqrt{\log n})  , \\
\qquad  \gamma\geq\frac{3}{2}; \\
\Theta(n^{3-\gamma} / (\log n)^{\frac{3}{2}}) , \\
 \qquad 1<\gamma<\frac{3}{2};  \\
\Theta(n^{2} / (\log n)^\frac{5}{2}) , \\
\qquad  \gamma=1;  \\
\Theta(n^{2} / (\log n)^\frac{3}{2})  , \\
\qquad  0\leq\gamma<1.\\
  \end{array}
\right. $ &
$\left\{
  \begin{array}{ll}
\Theta(n^{\frac{3}{2}})  , \\
\qquad  \gamma\geq\frac{3}{2}; \\
\Theta(n^{3-\gamma} / \log n) , \\
\qquad  1<\gamma<\frac{3}{2};  \\
\Theta(n^{2} / (\log n)^2) , \\
\qquad  \gamma=1;  \\
\Theta(n^{2} / \log n)  , \\
\qquad  0\leq\gamma<1.\\
  \end{array}
\right. $ \\
\hline

$0\leq\varphi<1$ & $\left\{
  \begin{array}{ll}
\Theta(n)  , \\
\qquad  \gamma>\frac{3}{2}; \\
\Theta(n \cdot \log n)  , \\
\qquad  \gamma=\frac{3}{2}; \\
\Theta(n^{\frac{5}{2}-\gamma}) , \\
 \qquad 1<\gamma<\frac{3}{2};  \\
\Theta(n^{\frac{3}{2}} / \log n) , \\
\qquad  \gamma=1;  \\
\Theta(n^{\frac{3}{2}})  , \\
\qquad  0\leq\gamma<1.\\
  \end{array}
\right. $ &
$\left\{
  \begin{array}{ll}
\Theta(n \cdot \log n)  , \\
\qquad  \gamma>\frac{3}{2}; \\
\Theta(n \cdot (\log n)^2)  , \\
\qquad  \gamma=\frac{3}{2}; \\
\Theta(n^{\frac{5}{2}-\gamma} \cdot \log n) , \\
\qquad  1<\gamma<\frac{3}{2};  \\
\Theta(n^{\frac{3}{2}}) , \\
\qquad  \gamma=1;  \\
\Theta(n^{\frac{3}{2}} \cdot \log n)  , \\
\qquad  0\leq\gamma<1.\\
  \end{array}
\right. $ &
$\left\{
  \begin{array}{ll}
\Theta(n^{2-\frac{\beta}{2}})  , \\
\qquad  \gamma>\frac{3}{2}; \\
\Theta(n^{2-\frac{\beta}{2}} \cdot \log n)  , \\
\qquad  \gamma=\frac{3}{2}; \\
\Theta(n^{\frac{7}{2}-\gamma-\frac{\beta}{2}}) , \\
\qquad 1<\gamma<\frac{3}{2};  \\
\Theta(n^{\frac{5}{2}-\frac{\beta}{2}} / \log n) , \\
\qquad  \gamma=1;  \\
\Theta(n^{\frac{5}{2}-\frac{\beta}{2}})  , \\
\qquad  0\leq\gamma<1.\\
  \end{array}
\right. $ &
$\left\{
  \begin{array}{ll}
\Theta(n^{\frac{3}{2}} / \sqrt{\log n})  , \\
\qquad  \gamma>\frac{3}{2}; \\
\Theta(n^{\frac{3}{2}} \cdot \sqrt{\log n})  , \\
\qquad  \gamma=\frac{3}{2}; \\
\Theta(n^{3-\gamma} / \sqrt{\log n}) , \\
\qquad 1<\gamma<\frac{3}{2};  \\
\Theta(n^{2} / (\log n)^{\frac{3}{2}}) , \\
\qquad  \gamma=1;  \\
\Theta(n^{2} / \sqrt{\log n})  , \\
\qquad  0\leq\gamma<1.\\
  \end{array}
\right. $ &
$\left\{
  \begin{array}{ll}
\Theta(n^{\frac{3}{2}})  , \\
\qquad  \gamma>\frac{3}{2}; \\
\Theta(n^{\frac{3}{2}} \cdot \log n)  , \\
\qquad  \gamma=\frac{3}{2}; \\
\Theta(n^{3-\gamma}) , \\
\qquad  1<\gamma<\frac{3}{2};  \\
\Theta(n^{2} / \log n) , \\
\qquad  \gamma=1;  \\
\Theta(n^{2})  , \\
\qquad  0\leq\gamma<1.\\
  \end{array}
\right. $ \\
\hline
\end{tabular}
}
\end{table*}
\begin{proof}
Let $ T_{l} $ denote the number of users with $ l $ destinations. Then, by laws of larger numbers (LLN, Lemma \ref{lem-Kolmogorov-slln} in Appendix \ref{apeneidisx-ausouful-lmemas}), we have
\begin{equation*}
  T_{l}  =  n\cdot\mathrm{Pr}(q_{k}=l)  = n\cdot{\left(\sum\nolimits_{j=1}^{n-1}j^{-\gamma}\right)^{-1}} \cdot (l)^{-\gamma}.
\end{equation*}
For all sessions $\{\mathbb{D}^{\mathrm{I}}_k\}_{k=1}^{n}$, we define two sets:
 \begin{equation*}
 \mathcal{K}^{1}:=\{k|q_{k}=\Theta(1)\}, \mathcal{K}^{\infty}:=\{k|q_{k}=\omega(1)\}.
\end{equation*}
Then, it follows that
\begin{equation}\label{Transport-sum}
\sum\nolimits_{k=1}^{n}|\EMST(\mathcal{P}^\mathrm{I}_k)|  = \widehat{{\Sigma}}^{1} + \widehat{{\Sigma}}^{\infty},
\end{equation}
where
 \begin{equation*}
\widehat{{\Sigma}}^{1} = \sum_{k\in\mathcal{K}^{1}}|\EMST(\mathcal{P}^\mathrm{I}_k)|,
\widehat{{\Sigma}}^{\infty} = \sum_{k\in\mathcal{K}^{\infty}}|\EMST(\mathcal{P}^\mathrm{I}_k)|.
\end{equation*}
We first address the part of $\widehat{{\Sigma}}^{1}$. Since for $q_{k}=\Theta(1)$, it holds that
$ |\EMST(\mathcal{P}^\mathrm{I}_k)|=\Theta(|X-v_{k}|)$,
we then have
$ \widehat{{\Sigma}}^{1}=\sum\nolimits_{k\in\mathcal{K}^{1}}|X-v_{k}|. $
 For $ k\in\mathcal{K}^{1} $, we define a sequence of random variables
$ \phi^{1}_{k}:=|X-v_{k}|/\sqrt{n} $,
which have finite meaning as follows:
$ E[\phi^{1}_{k}] = E[|X-v_{k}|]/\sqrt{n},$
where $E[|X-v_{k}|]$ is given in Lemma 5 of \cite{mobihoc2014}.
Then, according to Lemma \ref{lem-growthsteel}, with probability $1$, it holds that
\[ \widehat{{\Sigma}}^{1}=\Theta \left(\sqrt{n}\cdot\sum\nolimits_{k\in\mathcal{K}^{1}}\phi^{1}_{k}\right),\]
and
$
 \sum\nolimits_{k\in\mathcal{K}^{1}}\phi^{1}_{k} / |\mathcal{K}^{1}|=\Theta(E\left[|X-v_{k}|/\sqrt{n}\right]),
$
where $ |\mathcal{K}^{1}| $ denotes the cardinality of set $ \mathcal{K}^{1} $.
Therefore,
\[ \sum\nolimits_{k\in\mathcal{K}^{1}}\phi^{1}_{k} = \Theta\left(|\mathcal{K}^{1}|\cdot E\left[|X-v_{k}|/\sqrt{n}\right]\right).\]
Then, by Lemma \ref{lem-Kolmogorov-slln} in Appendix \ref{apeneidisx-ausouful-lmemas}, with probability $1$, it holds that
\begin{equation}\label{Transport-for-k-1}
 \widehat{{\Sigma}}^{1}=\Theta\left(|\mathcal{K}^{1}| \cdot E[|X-v_{k}|] \right).
\end{equation}
Next, we consider $ \widehat{{\Sigma}}^{\infty} $. By introducing anchor points, all random variables $ |\EMST(\mathcal{P}^\mathrm{I}_k)| $ with $k \in \mathcal{K}^{\infty}$ are independent. For users with follower number in $ \mathcal{K}^{\infty} $, we define two sets:
\begin{equation*}
\mathcal{K}^{\infty}_{1} = \{k|d_{k} = \Theta\left(1\right)\}, \mathcal{K}^{\infty}_{\infty} = \{k|d_{k} = \omega\left(1\right)\}.
\end{equation*}
Then, it follows that
\begin{equation}\label{Transport-for-k-wuqiong}
\widehat{{\Sigma}}^{\infty} =  \widehat{{\Sigma}}_{1}^{\infty} + \widehat{{\Sigma}}_{\infty}^{\infty},
\end{equation}
where
\begin{equation*}
\widehat{{\Sigma}}_{1}^{\infty} = \sum_{k\in\mathcal{K}^{\infty}_{1}}|\EMST \left(\mathcal{P}^\mathrm{I}_k\right)|,
\widehat{{\Sigma}}_{\infty}^{\infty} = \sum_{k\in\mathcal{K}^{\infty}_{\infty}}| \EMST \left(\mathcal{P}^\mathrm{I}_k\right)|.
\end{equation*}
We first consider the $ \mathcal{K}^{\infty}_{1} $. For $ d_{k} = \Theta\left(1\right) $, the order of $ \widehat{{\Sigma}}_{1}^{\infty} $ is lower than that of $ \widehat{{\Sigma}}^{1} $. For the final summation, the specific value of $ \widehat{{\Sigma}}_{1}^{\infty} $ is relatively infinitesimal.

Next, we consider $ \mathcal{K}^{\infty}_{\infty} $.
According to Lemma \ref{lem-EMST-B-D}, with probability $1$, it holds that
\begin{equation*}
|\mathrm{EMST\left(\mathcal{P}^\mathrm{I}_k\right)}| = \Theta\left(L^\mathrm{I}_{\mathcal{P}}(\beta, d_k)\right),
\end{equation*}
where $L^\mathrm{I}_{\mathcal{P}}(\beta, d_k)$ is defined in Eq.(\ref{eq-L-P-beta-9765}). And using LLN (Lemma \ref{lem-Kolmogorov-slln}), with probability $1$, it follows that
\begin{equation}\label{Transport-for-d-wuqiong}
\widehat{{\Sigma}}_{\infty}^{\infty} = \sum\nolimits_{l=2}^{n-1} \sum\nolimits_{d=1}^{l} T_{l} \cdot \mathrm{Pr}(d_{k}=d|q_{k}=l) \cdot L^\mathrm{I}_{\mathcal{P}}\left(\beta,d\right).
\end{equation}
Combining with Eq.(\ref{Transport-sum}), Eq.(\ref{Transport-for-k-1}), Eq.(\ref{Transport-for-k-wuqiong}) and Eq.(\ref{Transport-for-d-wuqiong}), we complete the proof of Theorem \ref{lem-delta-0-all-est}.
\end{proof}

\subsubsection{Tight Bounds on Transport Complexity}
Firstly, we give the main result (Theorem \ref{thm-tandsmport-cmomopelxsixyt-tihgoht-bounds}), i.e., the scaling laws of transport complexity, and prove them tight bounds on the transport complexity by deriving lower bounds (Lemma \ref{Z-D-thm-transport-load-on-sessions}) and computing upper bounds (Lemma \ref{dsfds-thm-upersdnmbpounds-on-sessions}), respectively.

\begin{thm}\label{thm-tandsmport-cmomopelxsixyt-tihgoht-bounds}
Let $\mathrm{C}_{\mathbb{N}}^{\mathrm{I}}$ denote the transport complexity
for Social-InterestCast sessions in a large-scale OSN over the carrier network with optimal communication architecture.
Then, it holds that
\begin{equation*}
  \mathrm{C}_{\mathbb{N}}^{\mathrm{I}}=\Theta\left(G\left(\beta, \gamma, \varphi\right)\right),
\end{equation*}
where $ G\left(\beta, \gamma, \varphi\right) $
depends on  the clustering exponents of relationship degree $\gamma$, relationship formation $\beta$ and dissemination pattern $ \varphi $, and the value
is presented in Table \ref{G-tab-total} in Appendix \ref{complete-main-reslts}.
 \end{thm}

Note that the results in Table \ref{G-tab-total} synthetically depend on three parameters, i.e., $\gamma$, $\beta$, and $\varphi$.
This contributes to the informative and comprehensive form of results.
For carding the flow and improving the readability, we move the detailed results to  Appendix \ref{complete-main-reslts} (Table \ref{G-tab-total}).
We will provide an intuitive explanation and discussion on the results in Section \ref{Z-D-subsub-explaination-social-dissemination-results}.

In the following proofs, we let $\mathrm{L}_{\mathbb{N}}^{\mathrm{I}}$ denote the transport complexity for all data dissemination sessions in OSN $\mathbb{N}$.

\textbf{Lower Bounds on Transport Complexity:}
The following Lemma \ref{Z-D-thm-transport-load-on-sessions} demonstrates a lower bound on transport complexity for OSN $\mathbb{N}$.
\begin{lem}\label{Z-D-thm-transport-load-on-sessions}
For the Social-InterestCast with the Zipf's distribution as defined in Eq.(\ref{zipf-final-destination}), it holds that
\begin{equation*}
  \mathrm{L}_{\mathbb{N}}^{\mathrm{I}}=\Omega\left(G\left(\beta, \gamma, \varphi\right)\right),
\end{equation*}
where $ G\left(\beta, \gamma, \varphi\right) $
is presented in Table \ref{G-tab-total} in Appendix \ref{complete-main-reslts}.
\end{lem}
%
%
%
%
%
%
%
\begin{proof}
Since
\begin{equation*}
\sum_{k=1}^{n} d_{k} = \Theta\left(\sum_{l=1}^{n-1} \sum_{d=1}^{l} T_{l} \cdot \mathrm{Pr}(d_{k}=d|q_{k}=l) \cdot d \right),
\end{equation*}
we can get that $ \sum\nolimits_{k=1}^{n} d_{k} = W\left(\gamma, \varphi\right),$ where $ W\left(\gamma, \varphi\right) $  is described in  Table \ref{tab-W-gamma-varphi}.
\begin{table}[!h]
\renewcommand{\arraystretch}{1.2}
\centering
\caption{The Number of All Destinations, $W\left(\gamma, \varphi\right)$}
\label{tab-W-gamma-varphi}
\scalebox{0.85}
{
\renewcommand\arraystretch{1.2}
\begin{tabular}{|p{2cm} ||p{5.9cm}|}
 \hline
$\varphi$ & $ W\left(\gamma, \varphi\right) $\\
\hline
\hline
$\varphi>2$ &
 $ \Theta(n) $  , $ \gamma\geq0.$ \\
\hline
$\varphi=2$ &
 $\left\{
  \begin{array}{ll}
\Theta(n)  , &  \gamma >1; \\
\Theta(n \cdot \log n)  , &  0\leq\gamma\leq1.
  \end{array}
\right. $ \\
\hline
$1<\varphi<2$ &
 $\left\{
  \begin{array}{ll}
\Theta(n) , &  \gamma>3-\varphi;  \\
\Theta(n \cdot \log n) , &  \gamma=3-\varphi; \\
\Theta(n^{4-\gamma-\varphi})  , &  1<\gamma<3-\varphi; \\
\Theta(n^{3-\varphi} / \log n) , &  \gamma=1; \\
\Theta(n^{3-\varphi})  , &  0\leq\gamma<1.
  \end{array}
\right. $ \\
\hline
$\varphi=1$ &
 $\left\{
  \begin{array}{ll}
\Theta(n)  , &  \gamma\geq2;\\
\Theta(n^{3-\gamma}/ \log n) , &  1<\gamma<2;  \\
\Theta(n^{2} / (\log n)^2) , &  \gamma=1;  \\
\Theta(n^{2}/ \log n)  , &  0\leq\gamma<1.
  \end{array}
\right. $ \\
\hline
$0\leq\varphi<1$ &
 $\left\{
  \begin{array}{ll}
\Theta(n)  , &  \gamma>2; \\
\Theta(n \cdot \log n)  , &  \gamma=2; \\
\Theta(n^{3-\gamma})  , &  1<\gamma<2; \\
\Theta(n^{2}/\log n) , &  \gamma=1;  \\
\Theta(n^{2})  , &  0\leq\gamma<1.
  \end{array}
\right. $ \\
\hline
\end{tabular}
}
\end{table}
Additionally, for all $ v_{k} \in \mathcal{V}$, 
\begin{equation*}
  E[| v_{k_{i}} - p_{k_{i}} |] = \Theta\left( \int_{0}^{\sqrt{n}} x \cdot e^{-\pi \cdot x^{2}}dx \right).
\end{equation*}
That is, $ E[| v_{k_{i}} - p_{k_{i}} |] = \Theta \left(1\right).$
Therefore, according to Lemma \ref{lem-Kolmogorov-slln}, with probability $1$, it follows that
\begin{equation}\label{Length-Anchor-Friends}
  \sum\nolimits_{k=1}^{n} \sum\nolimits_{i=1}^{d_{k}} |v_{k_{i}}-p_{k_{i}}| = \Theta \left(\sum\nolimits_{k=1}^{n} d_{k}\right) = \Theta\left(W\left(\gamma, \varphi\right)\right).
\end{equation}
Combining with Theorem \ref{lem-delta-0-all-est}, for all Social-InterestCast sessions $\mathbb{D}^{\mathrm{I}}_k$, $ k = 1, 2, ..., n $, with high probability, it holds that
\[ \sum\nolimits_{k=1}^n| \mathrm{EMST}(\mathbb{D}^{\mathrm{I}}_k) | = \Omega \left(G\left(\beta, \gamma, \varphi\right)\right) ,
\]
where $G\left(\beta, \gamma, \varphi\right)$ is provided in Table \ref{G-tab-total} in Appendix \ref{complete-main-reslts}.
Then, we complete the proof of Lemma \ref{Z-D-thm-transport-load-on-sessions}.
\end{proof}

\textbf{Upper Bounds on Transport Complexity:}
Here, we analyze the upper bounds on transport complexity
for Social-InterestCast sessions over an \emph{optimal} carrier network, i.e.,
a dedicated carrier communication network for this application.
In such a carrier network, a dedicated link can be built for every link in $\mathrm{EMST}(\mathcal{P}^\mathrm{I}_k)$, $k=1,2,\cdots, n$.
Therefore, from Lemma \ref{lem-EMST-B-D}  and Theorem \ref{lem-delta-0-all-est}, for Social-InterestCast sessions,
the aggregated transport distances of each session and all sessions can respectively reach to the orders as presented in  Eq.(\ref{eq-L-P-beta-9765}) and Table \ref{tab-lower-bound-EMST}. Then, we have
\begin{lem}\label{dsfds-thm-upersdnmbpounds-on-sessions}
For the Social-InterestCast with the Zipf's distribution as defined in Eq.(\ref{zipf-final-destination}), it holds that
\begin{equation*}
  \mathrm{L}_{\mathbb{N}}^{\mathrm{I}}=O\left(G\left(\beta, \gamma, \varphi\right)\right),
\end{equation*}
where $ G\left(\beta, \gamma, \varphi\right) $
is presented in Table \ref{G-tab-total} in Appendix \ref{complete-main-reslts}.
\end{lem}

\subsubsection{Explanation of Results}\label{Z-D-subsub-explaination-social-dissemination-results}
Based on the complete results in Table \ref{G-tab-total} in Appendix \ref{complete-main-reslts},
we can observe that
the final results vary in the range $\left[\Theta(n), \Theta(n^{2})\right]$, when the parameters $\gamma$, $\beta$, and $\varphi$ have different values in the range $\left[0, +\infty\right)$.

For Table \ref{G-tab-total}, it is too complex to provide a clear insight on the results.
To facilitate the understanding of our results, we choose a case of results and make them visualized. Figure \ref{fig-result-gamma3/2} illustrates the results for the case   $1<\beta<2$ and $1<\varphi<\frac{3}{2}$.

\begin{figure}[t]
\begin{center}
\scalebox{0.95}{\includegraphics{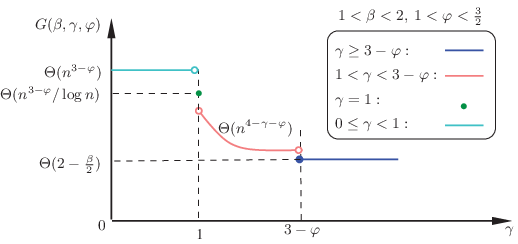}}
\caption{The Order of Lower Bounds on The Derived Transport
Complexity, $G\left(\beta, \gamma, \varphi\right)$, under $1<\beta<2$ and $1<\varphi<{3}/{2}$.} \label{fig-result-gamma3/2}
\end{center}
\vspace{-0.2in}
\end{figure}

Next, we mainly discuss the impacts of the clustering exponents of relationship degree $\gamma$, relationship formation $\beta$ and dissemination pattern $\varphi$ on the transport complexity.
Our result demonstrates that the transport complexity for Social-InterestCast is \emph{non-increasing}  within the range $\left[\Theta(n), \Theta(n^2)\right]$ in terms of the parameters  $\gamma$,  $\beta$, and  $\varphi$.
An intuitive explanation for this impact can be provided as follows:
Under each Social-InterestCast,
a larger clustering exponent of relationship degree $\gamma$ can limit the number of friends of each user into smaller upper bound with high probability, then leads to a lower transport complexity; a larger clustering exponent of relationship formation $\beta$ makes the followers more closer to each user with high probability, then possibly reduces the total transport distance of each Social-InterestCast session, finally also leads to a lower transport complexity;
a larger clustering exponent of dissemination pattern  $\varphi$ leads to a smaller probability that the source chooses larger number of destinations from its followers, which leads to a lower transport load.

\section{Conclusion and Future Work}\label{sec-conclud-future}
To measure the transport difficulty for data dissemination in online social networks (OSNs), we have defined a new metric, called transport complexity.
To model the formation of the interest-driven social session, i.e., Social-InterestCast,  we have proposed a four-layered architecture to model the data dissemination in OSNs, including the physical layer, social layer, content layer, and session layer.
By analyzing mutual relevances among these four layers, we have obtained the geographical distribution characteristics of dissemination sessions in OSNs. We have presented the density function of general social relationship distribution and the general form of bounds on transport load  for large-scale OSNs.
Furthermore, we have derived the tight  bounds on transport complexity of Social-InterestCast.

Much work still remains to be done. Main issues are listed as follows:

$\rhd$ In our work, although we have assumed that users are static, it is also applicable to the mobile scenario where
each mobile user moves within a bounded distance from exact one home point. Actually, in our numerical evaluation, we have made the most visited point of each mobile user as his static location. However, mobile users in real-life scenario are usually constrained by more than one home point rather than exact one home point. Therefore, it is necessary to further take into account a more realistic model for users distribution in the physical deployment layer, such as Multi-center Gaussian Model (MGM) in \cite{mgm2012}.

$\rhd$ We have provided only the explicit result for the model with homogeneous geographical distribution of users. This cannot still highlight sufficiently the characteristics of real-life OSNs or the advantages of the population-distance-based model.

$\rhd$ Under the profile \& social-based information dissemination pattern, a subsequent traffic session initiated from a source is often triggered by the previous session from another source. However, we have focused exclusively on the data arrival model without considering the correlations of data generating processes.

$\rhd$ When applying the metric \emph{transport complexity} to wireless broadcast, the definition will overestimate the transport difficulty for data dissemination in OSNs. To be specific, in some scenarios, data targeted to multiple destinations can be transmitted by a simple wireless broadcast, while our metric overly accumulate the transport distance of such a dissemination.
Even so, our results are still reasonable based on the explanation as follows: For a data dissemination, the distance of the last hop is relatively infinitesimal to the aggregate distance of this dissemination.
What's more, for the feature of scaling laws issue, we only care about the order of transport load imposed on the carrier communication networks of OSNs.
Anyway, it is a significant work to seek for a more accurate and practical metric to measure the transport difficulty of data dissemination in OSNs.

\bibliographystyle{abbrv}
\bibliography{osn}

\clearpage
\appendix

\section{Some Useful Lemmas}\label{apeneidisx-ausouful-lmemas}
We provide some useful lemmas as follows:
\begin{lem}[Minimal Spanning Tree \cite{steele1988growth}]\label{lem-growthsteel}
Let $X_i$, $1\leq i <\infty$, denote independent random variables with
values in $\mathbb{R}^d$, $d\geq 2$, and let $M_n$ denote the cost of
a minimal spanning tree of a complete graph with vertex set $\{X_i\}_{i=1}^n$, where the cost of an edge $(X_i, X_j)$ is given by $\Psi((|X_i-X_j|))$.
Here, $|X_i-X_j|$ denotes the Euclidean distance between $X_i$ and $X_j$ and
$\Psi$ is a monotone function. For bounded random variables and $0<\sigma<d$,
it holds that as $n\to \infty$, with probability $1$, one has
\[
M_n\sim c_1(\sigma,d)\cdot n^{\frac{d-\sigma}{d}} \cdot \int_{\mathbb{R}^d}f(X)^{\frac{d-\sigma}{d}}d X,\]
provided $\Psi(x)\sim x^\sigma$, where $f(X)$ is the density of the absolutely continuous part of the distribution of the $\{X_i\}$.
\end{lem}

\begin{lem}[Kolmogorov's Strong LLN \cite{williams1991probability}]\label{lem-Kolmogorov-slln}

Let $\{X_n\}$  be an i.i.d. sequence of random variables having finite mean: For $\forall n$,
$\mathbf{E}[X_n] < \infty$.
Then, a strong law of large numbers (LLN) applies to the sample mean:
\[\bar{X}_n \stackrel{a.s.}{\longrightarrow} \mathbf{E}[X_n],\]
where  $\stackrel{a.s.}{\longrightarrow}$ denotes \emph{almost sure convergence}.
\end{lem}

\section{Evaluations based on Gowalla Dataset and Brightkite dataset}\label{appendix-sec-evaluation}
In this section, we provide the evaluations of the adopted degree distribution model and  population-distance-based model using Gowalla and Brightkite users datasets \cite{Hossmann2011,Brightkite2007}.

\subsection{Gowalla and Brightkite Datasets}
Gowalla and Brightkite were both created in 2007. They were once two location-based social networking service providers where users shared their locations by ``checking-in" function, \cite{brightkite2010,Hossmann2011,Brightkite2007}. For Gowalla, the relationship network consists of $196,591$ nodes and $950,327$  undirected edges.
The Gowalla users' dataset in \cite{Hossmann2011} collected
a total of $6,442,890$ checkins of these users over the period from February 2009 to October 2010.
It provides each user with the incoming  and outgoing follower lists as well as the latitude and longitude.
For Brightkite, the relationship network consists of $58,228$ nodes and $214,078$ directed edges.
The Brightkite users' dataset in \cite{Brightkite2007} collected a total of $4,491,143$ checkins of these users over the period from April 2008 to October 2010. It provides each user with the incoming and outgoing follower lists.

In our evaluations, because of the deficiency of users' data in Asia and other areas, for  Gowalla users' dataset \cite{Hossmann2011}, we extracted $52,161$ users who locate in North America to improve the accuracy and decrease the computation complexity.
Figure \ref{fig-location-distribution} shows the geographical distribution of users in North America. $Y$ and $X$ represent the latitude and longitude of users' locations, respectively.

\begin{figure}[!h]
\begin{center}
\scalebox{0.55}{\includegraphics{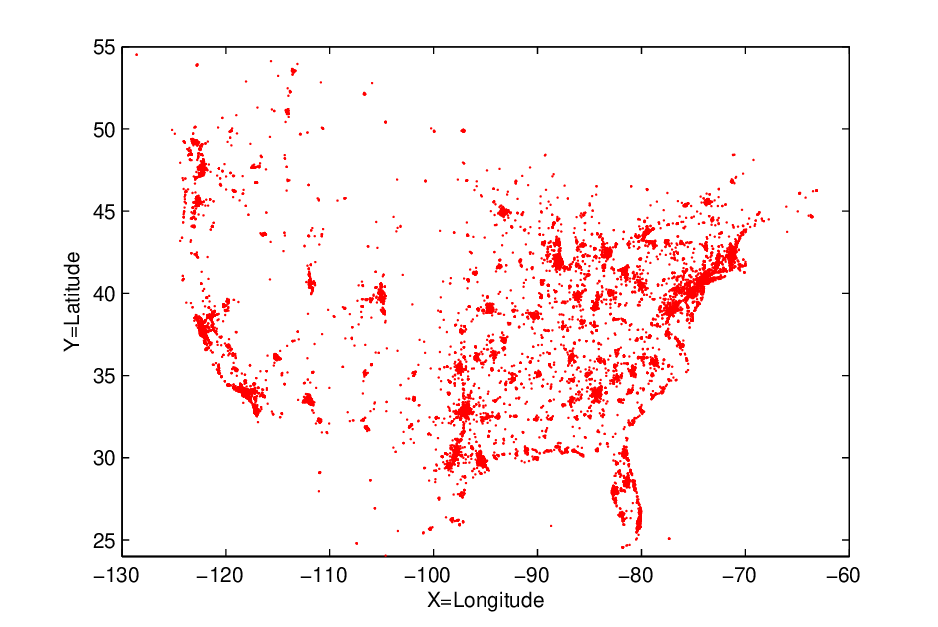}}
\caption{The geographical position of Gowalla Users in North America.}\label{fig-location-distribution}
\end{center}
\end{figure}

\subsection{Evaluation of Degree Distribution}\label{appendix-ddobu234324}
In this work, we assume that the number
of followers of a particular node $v_k\in \mathcal{V}$, denoted by $q_k$, follows a Zipf's
distribution \cite{manning1999foundations}, i.e.,
\[
\Pr(q_k=l)={\left(\sum\nolimits_{j=1}^{n-1}j^{-\gamma}\right)^{-1}} \cdot {l^{-\gamma}}.
\]
Next, we validate the Zipf's degree distribution of social relationships by investigating the  negative linear correlation between
\begin{center}
$Y:=  N_{out}(K_{out})$ and $X:=  K_{out}$,
\end{center}
 where
$K_{out}$ represents an outgoing degree, and $N_{out}(K_{out})$ denotes the number of the users with an outgoing degree $K_{out}$.

\begin{figure*}[!t]
\begin{flushleft}
\begin{tabular}{cc}
\scalebox{0.6}{\includegraphics{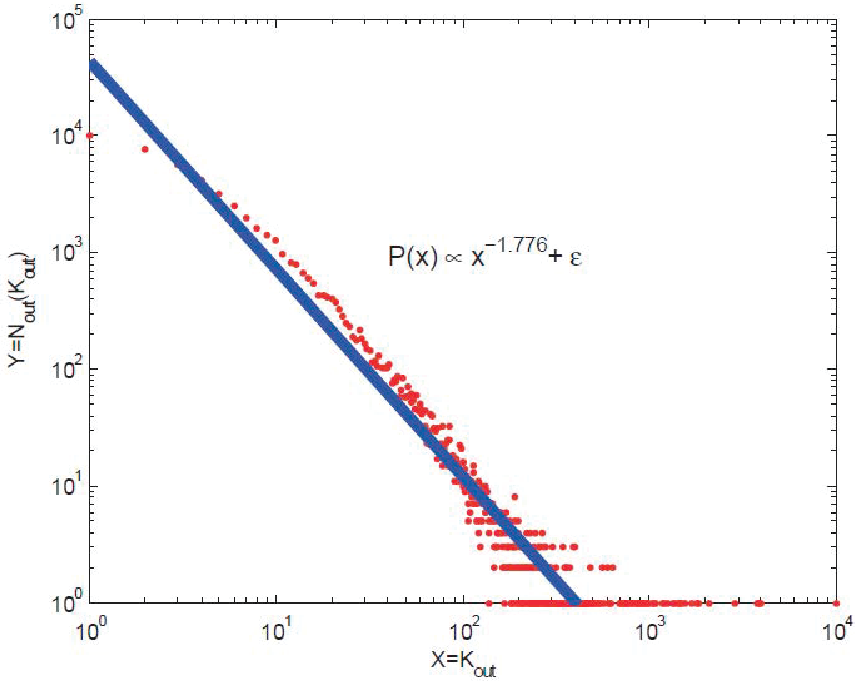}}&
\scalebox{0.6}{\includegraphics{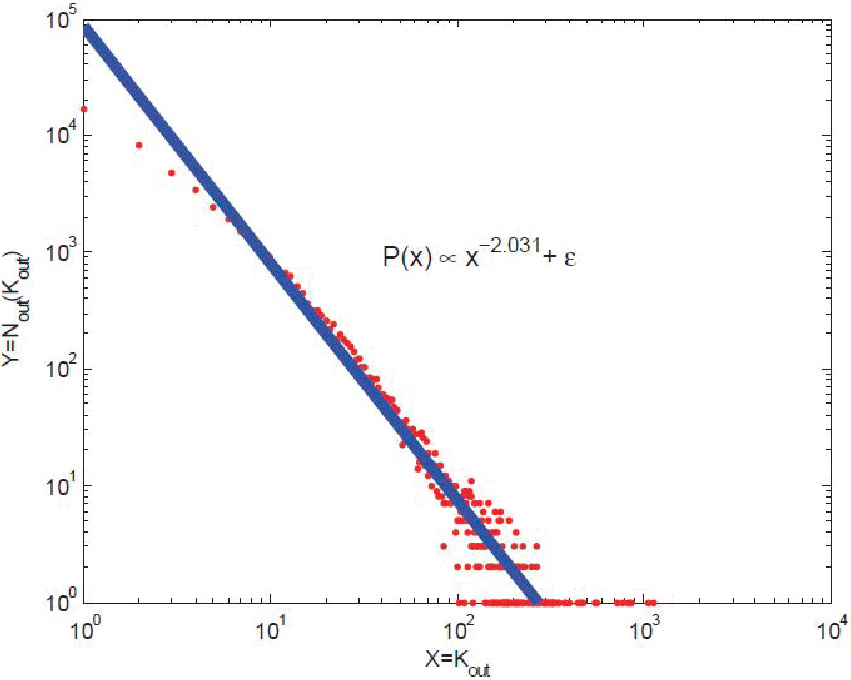}}
 \\
(a)  & (b)
\end{tabular}
\end{flushleft}
\caption{(a) is the validation for social degree distribution of Gowalla users, while (b) is of Brightkite users from \cite{mobihoc2014}.} \label{fig-evaluation-outdegree}
 \end{figure*}

In Gowalla and Brightkite datasets \cite{Hossmann2011,Brightkite2007}, the correlation between $Y$ and $X$ is described as in Figure \ref{fig-evaluation-outdegree}.
It shows that the correlation is approximately to  be a line segment with negative slope, which basically matches our proposed model.

\subsection{Evaluation of Population-Distance-Based Model}\label{appendix-vopbm23432432}
We discretize the network area $\mathcal{O}$ into a lattice consisting of  $120,000$ points.
 Each point acts as a candidate \emph{anchor point}.
We denote this lattice  by $\mathcal{O}_d$.

Let $d(u,p)$ denote the distance between user $u$ and a random position/cell $p$ in $\mathcal{O}_d$ that serves as a candidate \emph{anchor point}. Let $\mathcal{D}(u,p)$
denote the disk centered at $u$ with a radius $d(u,p)$. Let
$N(u,p)$ denote the number of nodes in the disk $\mathcal{D}(u,p)$. Let $v_p$ denote  the closest user to the candidate \emph{anchor point} $p$.
Furthermore, we define a variable \[\mathbf{I}(u,v_p)=\textbf{1}\cdot  \{ v_p\mbox{ is a follower of } u \}.\]

\begin{figure*}[!t]
\begin{flushleft}
\begin{tabular}{cc}
\scalebox{0.355}{\includegraphics{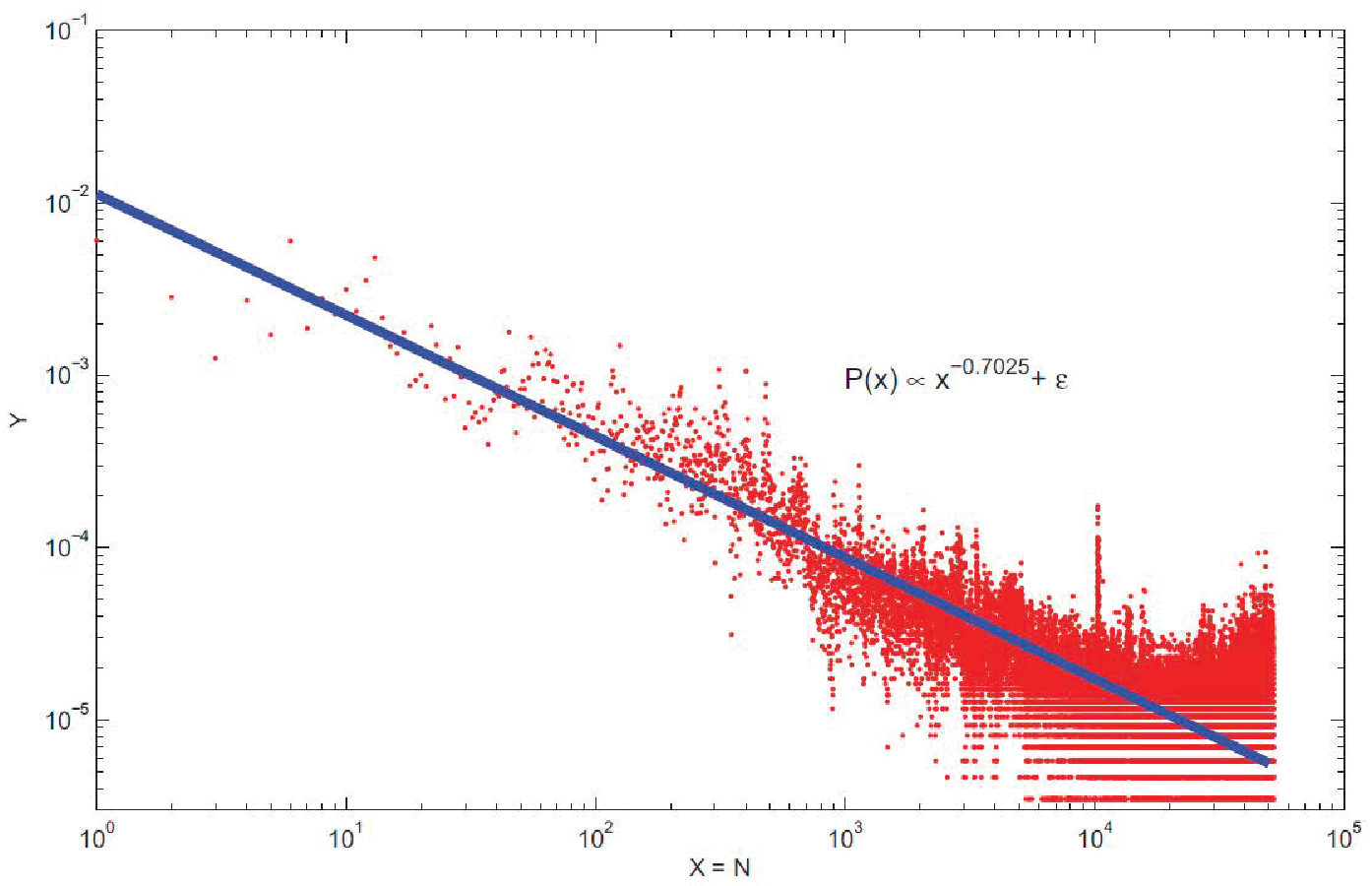}}&
\scalebox{0.35}{\includegraphics{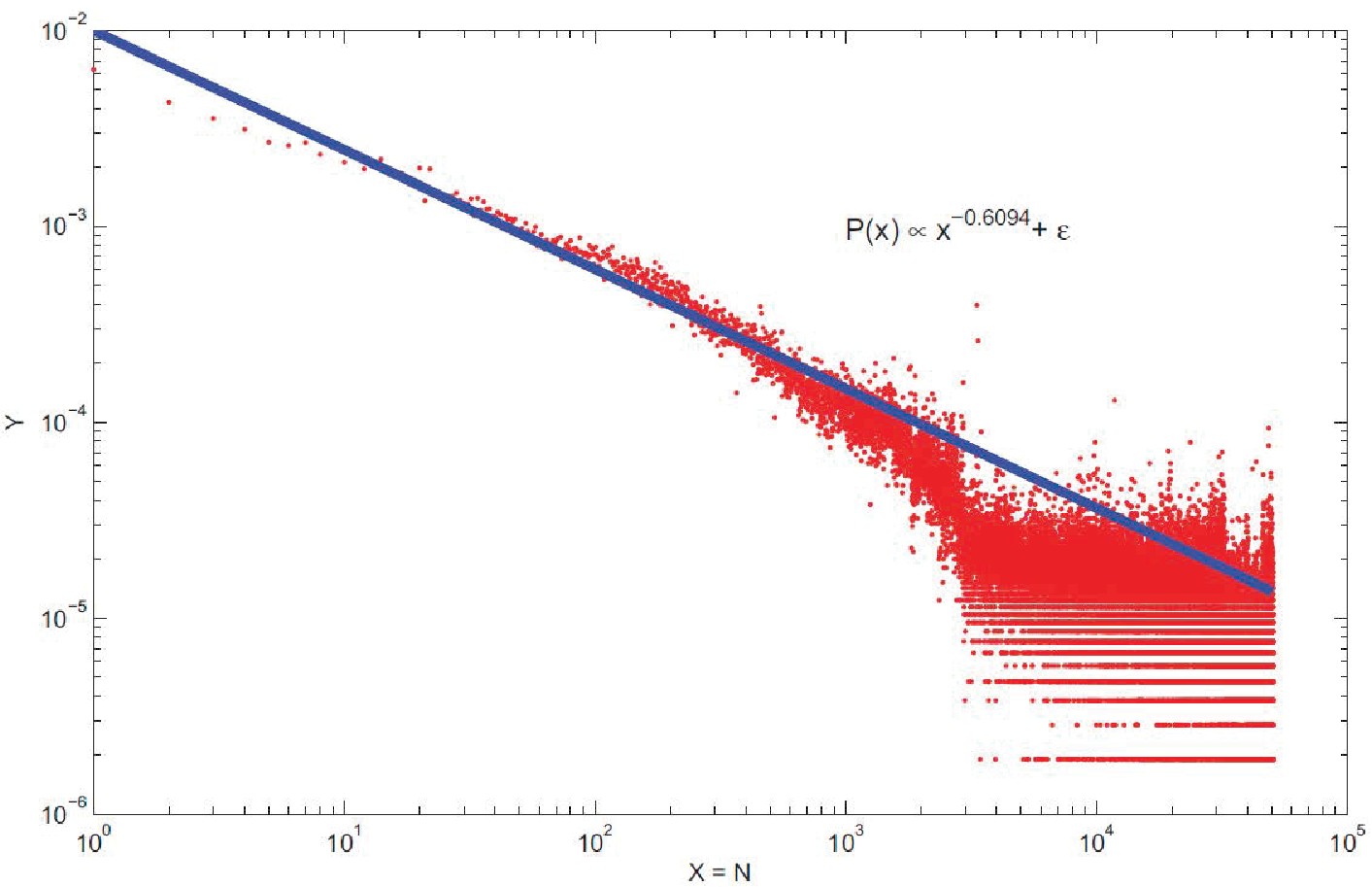}}
 \\
(a)  & (b)
\vspace{-0.2in}
\end{tabular}
\end{flushleft}
\caption{(a) is the validation for Population-Distance-Based social probability distribution of Gowalla users, while (b) is of Brightkite users.
 } \label{fig-evaluation-formadistri}
 \vspace{-0.1in}
 \end{figure*}

We validate the geographical distribution of relationships by investigating the \emph{negative linear correlation} between $Y$ and $X$, where
$X$ denotes a number of nodes contained in a certain disk, and
$$
Y:=  \frac{1}{|\mathcal{E}| }\cdot{\sum\limits_{ \mathbf{I}(u,v_p)=1} \textbf{1}\cdot  \{ N(u,p)=N \}} ,
$$
with $\mathcal{E}$ denoting  the set of all social links, respectively.


In real-world dataset, the candidate \emph{anchor points} located on the sea or in the desert are quite far from their nearest users, which leads to a high probability of outer sphere of users chosen to be a follower of $v_k$.
To get rid of these candidate \emph{anchor points}  which are apart from the corresponding users, we set a threshold distance $d_f$ to filter such positions as outliers. In this Gowalla dataset \cite{Hossmann2011}, $d_f$ is set to be $200$ kilometers, which makes the positions $p$ cover most of the land and filter the ocean area simultaneously.


In the Gowalla and Brightkite datasets  \cite{Hossmann2011,Brightkite2007}, the correlation between $Y$ and $X$ is described by
Figure \ref{fig-evaluation-formadistri}.
It shows that the correlation tendency is approximated very coarsely to a line segment with a negative slope.

The experimental result also basically validates the model, although it does not perfectly match.
The main reason of mismatch possibly lies in the facts as follows:
(1) The locations of users in these datasets are actually the positions where they check-in, instead of the place where they usually stay.
(2) Based on these dataset, more than $90\%$ results fall within the cases with $X>10^{3}$. The accumulation of experimental errors here leads to the ``bloated" tails in the evaluation figures.

\section{Complete Main Results}\label{complete-main-reslts}
Here, we summarize the main results of a complete form in the following table:

\begin{landscape}
\begin{table}[!l]
\renewcommand{\arraystretch}{1.8}
\caption{The Order of Transport Complexity,  $G\left(\beta, \gamma, \varphi\right)$}
\label{G-tab-total}
\scalebox{0.8}
{
\begin{tabular}{|p{1.8cm}||p{4.7cm}|p{4.7cm}|p{4.7cm}|p{4.7cm}|p{4.7cm}|}

\hline
$\varphi$ $\backslash$ $\beta$ & $\beta>2$ & $\beta=2$ & $ 1<\beta<2 $ & $\beta=1$ & $0\leq\beta<1$ \\
\hline
\hline

$\varphi>2$ &
 $ \Theta(n) , ~~ \gamma\geq0 $ &
$ \Theta(n \cdot \log n) , ~~ \gamma\geq0 $ &
$ \Theta(n^{2-\frac{\beta}{2}}) , ~~ \gamma\geq0 $ &
$ \Theta(n^{\frac{3}{2}} / \sqrt{\log n}), ~~ \gamma\geq0 $ &
$ \Theta(n^{\frac{3}{2}}) , ~~ \gamma\geq0 $ \\
\hline

$\varphi=2$ & $\left\{
  \begin{array}{ll}
\Theta(n)  , ~~  \gamma >1;\\
\Theta(n \cdot \log n) , ~~ 0\leq\gamma\leq1.\\
  \end{array}
\right. $ &
 $ \Theta(n \cdot \log n) $  , ~~ $ \gamma\geq0. $  &
 $ \Theta(n^{2-\frac{\beta}{2}})  , ~~ \gamma\geq0. $  &
 $ \Theta(n^{\frac{3}{2}} / \sqrt{\log n}) $ , ~~ $ \gamma\geq0. $ &
 $ \Theta(n^{\frac{3}{2}})  , ~~ \gamma\geq0. $  \\
\hline

$\frac{3}{2}<\varphi<2$ & $\left\{
  \begin{array}{ll}
\Theta(n) , ~~ \gamma>3-\varphi;  \\
\Theta(n \cdot \log n) , ~~  \gamma=3-\varphi; \\
\Theta(n^{4-\gamma-\varphi})  , ~~ 1<\gamma<3-\varphi; \\
\Theta(n^{3-\varphi} / \log n) ,  ~~ \gamma=1; \\
\Theta(n^{3-\varphi})  , ~~ 0\leq\gamma<1.
  \end{array}
\right. $ &
$\left\{
  \begin{array}{ll}
\Theta(n \cdot \log n) , ~~ \gamma\geq3-\varphi; \\
\Theta(n^{4-\gamma-\varphi})  , ~~ 1<\gamma<3-\varphi; \\
\Theta(n^{3-\varphi} / \log n) , ~~ \gamma=1; \\
\Theta(n^{3-\varphi})  , ~~ 0\leq\gamma<1.
  \end{array}
\right. $ &
$
\Theta(n^{2-\frac{\beta}{2}}) , ~~ \gamma\geq0. $ &
$\Theta(n^{\frac{3}{2}} / \sqrt{\log n}) , ~~  \gamma\geq0. $ &
$ \Theta(n^{\frac{3}{2}}) , ~~ \gamma\geq0. $ \\
\hline

$\varphi=\frac{3}{2}$ & $\left\{
  \begin{array}{ll}
\Theta(n)  ,  ~~ \gamma >\frac{3}{2};\\
\Theta(n \cdot \log n) , ~~ \gamma=\frac{3}{2};  \\
\Theta(n^{\frac{5}{2}-\gamma}) , ~~ 1<\gamma<\frac{3}{2};  \\
\Theta(n^{\frac{3}{2}} / \log n) , ~~  \gamma=1;  \\
\Theta(n^{\frac{3}{2}})  ,  ~~ 0\leq\gamma<1.\\
  \end{array}
\right. $ &
$\left\{
  \begin{array}{ll}
\Theta(n \cdot \log n)  ,  ~~  \gamma \geq\frac{3}{2};\\
\Theta(n^{\frac{5}{2}-\gamma}) , ~~ 1<\gamma<\frac{3}{2};  \\
\Theta(n^{\frac{3}{2}} / \log n) , ~~ \gamma=1;  \\
\Theta(n^{\frac{3}{2}})  , ~~  0\leq\gamma<1.\\
  \end{array}
\right. $ &
$\left\{
  \begin{array}{ll}
\Theta(n^{2-\frac{\beta}{2}})  ,  ~~   \gamma \geq\frac{3}{2};\\
\Theta(n^{\frac{5}{2}-\gamma}) ,  ~~ 1<\gamma<\frac{3}{2};  \\
\Theta(n^{\frac{3}{2}} / \log n) , ~~ \gamma=1;  \\
\Theta(n^{\frac{3}{2}})  ,  ~~ 0\leq\gamma<1.\\
  \end{array}
\right. $ &
$\left\{
  \begin{array}{ll}
\Theta(n^{\frac{3}{2}} / \sqrt{\log n})  ,  ~~   \gamma >1;\\
\Theta(n^{\frac{3}{2}} \cdot \sqrt{\log n})  , ~~  0\leq\gamma\leq1.\\
  \end{array}
\right. $ &
$\left\{
  \begin{array}{ll}
\Theta(n^{\frac{3}{2}})  , ~~   \gamma >1;\\
\Theta(n^{\frac{3}{2}} \cdot \log n)  , ~~  0\leq\gamma\leq1.\\
  \end{array}
\right. $ \\
\hline

$1<\varphi<\frac{3}{2}$ & $\left\{
  \begin{array}{ll}
\Theta(n)  , ~~   \gamma>3-\varphi; \\
\Theta(n \cdot \log n) , ~~ \gamma=3-\varphi; \\
\Theta(n^{4-\gamma-\varphi}) , ~~  1<\gamma<3-\varphi;  \\
\Theta(n^{3-\varphi} / \log n) , ~~ \gamma=1; \\
\Theta(n^{3-\varphi})  ,  ~~ 0\leq\gamma<1.\\
  \end{array}
\right. $ &
$\left\{
  \begin{array}{ll}
\Theta(n \cdot \log n)  , ~~  \gamma\geq3-\varphi; \\
\Theta(n^{4-\gamma-\varphi}) , ~~ 1<\gamma<3-\varphi;  \\
\Theta(n^{3-\varphi} / \log n) , ~~ \gamma=1; \\
\Theta(n^{3-\varphi})  , ~~  0\leq\gamma<1.\\
  \end{array}
\right. $ &
$\left\{
  \begin{array}{ll}
\Theta(n^{2-\frac{\beta}{2}})  , ~~   \gamma\geq3-\varphi; \\
\Theta(n^{4-\gamma-\varphi}) , ~~  1<\gamma<3-\varphi;  \\
\Theta(n^{3-\varphi} / \log n) , ~~  \gamma=1; \\
\Theta(n^{3-\varphi})  , ~~ 0\leq\gamma<1.\\
  \end{array}
\right. $ &
$\left\{
  \begin{array}{ll}
\Theta(n^{\frac{3}{2}} / \sqrt{\log n})  ,  ~~ \gamma>\frac{5}{2}-\varphi; \\
\Theta(n^{\frac{3}{2}} \cdot  \sqrt{\log n}) , ~~  \gamma=\frac{5}{2}-\varphi; \\
\Theta(n^{4-\gamma-\varphi}) , ~~ 1<\gamma<\frac{5}{2}-\varphi;  \\
\Theta(n^{3-\varphi} / \log n) , ~~ \gamma=1; \\
\Theta(n^{3-\varphi})  ,  ~~ 0\leq\gamma<1.\\
  \end{array}
\right. $ &
$\left\{
  \begin{array}{ll}
\Theta(n^{\frac{3}{2}})  ,  ~~ \gamma>\frac{5}{2}-\varphi; \\
\Theta(n^{\frac{3}{2}} \cdot \log n) , ~~ \gamma=\frac{5}{2}-\varphi; \\
\Theta(n^{4-\gamma-\varphi}) ,  ~~ 1<\gamma<\frac{5}{2}-\varphi;  \\
\Theta(n^{3-\varphi} / \log n) , ~~ \gamma=1; \\
\Theta(n^{3-\varphi})  ,  ~~  0\leq\gamma<1.\\
  \end{array}
\right. $ \\
\hline

$\varphi=1$ & $\left\{
  \begin{array}{ll}
\Theta(n)  ,  ~~  \gamma\geq 2;\\
\Theta(n^{3-\gamma} / \log n) , ~~ 1<\gamma<2;  \\
\Theta(n^{2} / (\log n)^2) , ~~   \gamma=1;  \\
\Theta(n^{2} / \log n)  , ~~  0\leq\gamma<1.\\
  \end{array}
\right. $ &
$\left\{
  \begin{array}{ll}
\Theta(n \cdot \log n)  ,  ~~  \gamma\geq 2;\\
\Theta(n^{3-\gamma} / \log n) , ~~ 1<\gamma<2;  \\
\Theta(n^{2} / (\log n)^2) , ~~  \gamma=1;  \\
\Theta(n^{2} / \log n)  , ~~  0\leq\gamma<1.\\
  \end{array}
\right. $ &
$\left\{
  \begin{array}{ll}
\Theta(n^{2-\frac{\beta}{2}})  , ~~   \gamma>\frac{3}{2};\\
\Theta(n^{3-\gamma} / \log n) , ~~ 1<\gamma\leq\frac{3}{2};  \\
\Theta(n^{2} / (\log n)^2) , ~~  \gamma=1;  \\
\Theta(n^{2} / \log n)  , ~~  0\leq\gamma<1.\\
  \end{array}
\right. $ &
$\left\{
  \begin{array}{ll}
\Theta(n^{\frac{3}{2}} / \sqrt{\log n})  , ~~   \gamma\geq \frac{3}{2};\\
\Theta(n^{3-\gamma} / \log n) ,  ~~ 1<\gamma<\frac{3}{2};  \\
\Theta(n^{2} / (\log n)^2) , ~~  \gamma=1;  \\
\Theta(n^{2} / \log n)  , ~~  0\leq\gamma<1.\\

  \end{array}
\right. $ &
$\left\{
  \begin{array}{ll}
\Theta(n^{\frac{3}{2}})  , ~~   \gamma\geq \frac{3}{2};\\
\Theta(n^{3-\gamma} / \log n) , ~~  1<\gamma<\frac{3}{2};  \\
\Theta(n^{2} / (\log n)^2) , ~~  \gamma=1;  \\
\Theta(n^{2} / \log n)  ,  ~~  0\leq\gamma<1.\\
  \end{array}
\right. $ \\
\hline

$0\leq\varphi<1$ & $\left\{
  \begin{array}{ll}
\Theta(n)  ,  \gamma>2;\\
\Theta(n \cdot \log n) , ~~  \gamma=2;  \\
\Theta(n^{3-\gamma}) , ~~  1<\gamma<2;  \\
\Theta(n^{2} / \log n) , ~~  \gamma=1;  \\
\Theta(n^{2})  , ~~  0\leq\gamma<1.\\
  \end{array}
\right. $ &
$\left\{
  \begin{array}{ll}
\Theta(n \cdot \log n)  ,  ~~   \gamma\geq2;\\
\Theta(n^{3-\gamma}) , ~~  1<\gamma<2;  \\
\Theta(n^{2} / \log n) , ~~  \gamma=1;  \\
\Theta(n^{2})  ,  ~~  0\leq\gamma<1.\\
  \end{array}
\right. $ &
$\left\{
  \begin{array}{ll}
\Theta(2-\frac{\beta}{2})  , ~~  \gamma\geq2;\\
\Theta(n^{3-\gamma}) , ~~  1<\gamma<2;  \\
\Theta(n^{2} / \log n) , ~~  \gamma=1;  \\
\Theta(n^{2})  ,  ~~  0\leq\gamma<1.\\
  \end{array}
\right. $ &
$\left\{
  \begin{array}{ll}
\Theta(n^{\frac{3}{2}} / \sqrt{\log n})  , ~~ \gamma>\frac{3}{2};\\
\Theta(n^{\frac{3}{2}} \cdot \sqrt{\log n}) ,  ~~ \gamma=\frac{3}{2};  \\
\Theta(n^{3-\gamma}) , ~~ 1<\gamma<\frac{3}{2};  \\
\Theta(n^{2} / \log n) ,  ~~ \gamma=1;  \\
\Theta(n^{2})  ,  ~~  0\leq\gamma<1.\\
  \end{array}
\right. $ &
$\left\{
  \begin{array}{ll}
\Theta(n^{\frac{3}{2}})  ,  ~~  \gamma>\frac{3}{2};\\
\Theta(n^{\frac{3}{2}} \cdot \log n) , ~~ \gamma=\frac{3}{2};  \\
\Theta(n^{3-\gamma}) ,  ~~ 1<\gamma<\frac{3}{2};  \\
\Theta(n^{2} / \log n) ,  ~~ \gamma=1;  \\
\Theta(n^{2})  ,  ~~  0\leq\gamma<1.\\
  \end{array}
\right. $ \\
\hline
\end{tabular}
}
\end{table}
\end{landscape}

\end{document}